\documentclass{article}

\usepackage[final]{neurips_2021}

\usepackage[utf8]{inputenc} 
\usepackage[T1]{fontenc}    
\usepackage{hyperref}       
\usepackage{url}            
\usepackage{booktabs}       
\usepackage{amsfonts}       
\usepackage{nicefrac}       
\usepackage{microtype}      
\usepackage{xcolor}         

\usepackage{amsmath}
\usepackage{cryptocode}

\usepackage{thmtools,thm-restate}
\usepackage{comment}

\usepackage{thmtools}
\usepackage{mathtools}
\usepackage{enumerate}

\usepackage{amsfonts}
\usepackage{amsthm}

\usepackage{amssymb}
\usepackage{amsmath}

\usepackage{algorithm, algorithmic}

\usepackage{dirtytalk}

\usepackage{subfig}
\usepackage{booktabs}
\usepackage{siunitx}

\usepackage{bbm}

\usepackage{wrapfig}

\usepackage{tikz}
\usetikzlibrary{quantikz}
 
\newcommand{\N}{\mathbb{N}}

\newcommand{\R}{\mathbb{R}}
\newcommand{\C}{\mathbb{C}}

\newcommand{\compcirc}{G}

\newcommand{\inspace}{{\mathcal X}}
\newcommand{\outspace}{{\mathcal Y}}

\newtheorem{theorem}{Theorem}
\newtheorem{lemma}{Lemma}

\newtheorem{corollary}{Corollary}

\newtheorem{fact}{Fact}

\newtheorem{note}{Note}

\newtheorem{definition}{Definition}

\newcommand{\e}{\epsilon}
\newcommand{\nbits}{\{0,1\}^n}

\newcommand{\poly}{\text{poly}}
\newcommand{\oracle}{\mathtt{O}}
\newcommand{\dist}{\mathcal{D}}
\newcommand{\state}[1]{|#1 \rangle}
\newcommand{\id}{\mathbb{I}}
\newcommand{\adv}{\mathbf{A}}
\newcommand{\x}{\mathbf{x}}
\newcommand{\Prob}{\mathbb{P}}
\newcommand{\ver}{\mathbf{V}}
\newcommand{\prov}{\mathbf{P}}
\newcommand{\comp}{G}
\newcommand{\heli}{d_H}
\newcommand{\tv}{\triangle}
\newcommand{\clock}{\text{clock}}
\newcommand{\out}{\text{out}}
\newcommand{\aux}{\text{aux}}
\newcommand{\adver}{\text{adv}}
\newcommand{\inp}{\text{in}}
\newcommand{\prop}{\text{prop}}
\newcommand{\comput}{\text{comp}}
\newcommand{\subspace}{\mathcal{L}}
\newcommand{\wt}{\widetilde}
\newcommand{\tl}{\tilde}
\newcommand{\trace}{\text{Tr}}
\newcommand{\hist}{\text{hist}}

\newlength\myindent
\usepackage{todonotes}

\newcounter{mycomment}

\title{Provable Adversarial Robustness in the Quantum Model}

\author{%
  Khashayar Barooti \\
  EPFL\\
  Lausanne, Switzerland\\
  \texttt{khashayar.barooti@epfl.ch} \\
  \And
   Grzegorz Głuch \\
   EPFL \\
   Lausanne, Switzerland \\
   \texttt{grzegorz.gluch@epfl.ch}  \\
   \AND
    Ruediger Urbanke \\
    EPFL \\
    Lausanne, Switzerland \\
    \texttt{ruediger.urbanke@epfl.ch}  \\
}

\begin{document}

\maketitle

\begin{abstract}
Modern machine learning systems have been applied successfully to a variety of tasks in recent years but making such systems robust against adversarially chosen modifications of input instances seems to be a much harder problem. It is probably fair to say that no fully satisfying solution has been found up to date and it is not clear if the standard formulation even allows for a principled solution. Hence, rather than following the classical path of bounded perturbations, we consider a model similar to the quantum PAC-learning model introduced by Bshouty and Jackson [1995]. Our first key contribution shows that in this model we can reduce adversarial robustness to the conjunction of two classical learning theory problems, namely (Problem 1) the problem of finding generative models and (Problem 2) the problem of devising classifiers that are robust with respect to distributional shifts. Our second key contribution is that the considered framework does not rely on specific (and hence also somewhat arbitrary) threat models like $\ell_p$ bounded perturbations. Instead, our reduction guarantees that in order to solve the adversarial robustness problem in our model it suffices to consider a single distance notion, i.e. the Hellinger distance. From the technical perspective our protocols are heavily based on the recent advances on delegation of quantum computation, e.g. Mahadev [2018]. 

Although the considered model is quantum and therefore not immediately applicable to ``real-world'' situations, one might hope that in the future either one can find a way to embed ``real-world'' problems into a quantum framework or that classical algorithms can be found that are capable of mimicking their powerful quantum counterparts.
\end{abstract}

\section{Introduction}\label{sec:intro}

Modern machine learning systems are known to be susceptible to small, adversarially-chosen perturbations of the inputs \citep{intriguingprop,neuralnetseasilyfooled}. Even though a considerable number of papers has been published on this subject, and much progress has been made on understanding the underlying problem, it is probably fair to say that a satisfying solution has not been found yet. Indeed, at present it is not clear if the problem, as traditionally stated, is even tractable.
We introduce a new model for adversarial robustness that crucially makes use of properties of quantum mechanics. For this model we give a protocol that provably guarantees robustness. 

The standard approach to adversarial machine learning is to consider perturbation sets. These sets describe by how much the adversary is allowed to perturbe the input according to our model. The most commonly considered such perturbation sets describe perturbations that are bounded in $\ell_p$ norms \citep{certifiableRobust1, certifiableRobust2}, but other perturbations (rotations, shifts, etc.) were also considered \citep{spatialpert}. To date there is a considerable literature on this approach. 

Whether those assumed bounds capture real-world scenarios is of course up for debate. E.g., it has been shown that defending models against one perturbation set does not necessarily improve the robustness against other perturbations and that there exist a trade-off between robustness for different perturbations \citep{trameroverfitting}.

Perhaps more importantly, it has been shown that it might be impossible to find robust classifiers, even in the restrictive model of $\ell_p$ bounded perturbations. For instance there exist tasks for which it is easy to find high-accuracy classifiers but finding robust models is infeasible. E.g., in \citet{computationalhardness} the authors describe a situation where it is information-theoretically easy to learn robustly but there is no algorithm in the statistical query model that computes a robust classifier. 
In \citet{computationalhardnessfull} an even stronger result is proven. It is shown that under a standard cryptographic assumption there exist learning tasks for which no algorithm can efficiently learn a robust classifier. Finally, in \citet{robustnessaccuracy} it was shown that there might be an inherent tension between robust and accurate classifiers. All these negative results were considered in the white-box model. The situation looks more hopeful when one considers the black-box model. For instance, in \citet{querycomplexity} it is shown that it still might be possible to construct robust classifiers in the black-box model despite the previous negative results. Unfortunately this result heavily relies on the fact that the perturbation set is fixed.

Some of the literature tries to escape the assumption of $\ell_p$ bounded perturbations or fixed perturbation sets altogether. E.g., in \citet{srebrounknown} it is shown how to defend against adversaries that are allowed to use perturbations from a set that is not known to the defender. Unfortunately, there are cases where the defense requires exponentially (in the VC-dimension) many samples. Another approach was considered in
\citet{Goldwasser}. The authors show how to get rid of the limitations on perturbations completely. This, naturally, comes at a price: e.g., the results in \citet{Goldwasser} are based on the assumption that the learner can decide not to give an answer for some inputs (selective learning \citep{selectiveclassifiers}) and that they see the test set upfront (transductive learning), thus the capabilities of the learner are enhanced. A recent work by \citet{gluchSeparation} proposes another approach for removing restrictions on perturbation sets. In this paper it is shown that an effective adversary can be used to design a faster learning algorithm. If one assumes that the learning problem is hard, this approach leads to a defense. The limitation of this approach is that the adversary is required to generate adversarial examples at random.

From the discussion above it seems that improving the state of the art of robustness for $\ell_p$ bounded perturbations might not guarantee a satisfying solution to the adversarial robustness puzzle in the long run. Instead, an exploration of new models is likely needed for a principled resolution of the problem. This is the path we follow in this work. We propose a new model and then show how to design a provable defense for it.

The model we propose crucially uses properties of quantum mechanics. Instead of the standard samples $x \sim \dist$ we assume that the adversary has access to the quantum states $\sum_{x \in \{0,1\}^n} \sqrt{\dist(x)}\ket{x}$ -- similar  to the quantum PAC-learning model by \citet{quantumPAC}. This model gives the adversary more power so, intuitively, he can trick us in more ways but is also able to execute more complicated strategies. A priori it is not clear how the two models relate. We will show that the quantum model allows to certify more properties than the classical counterpart.

It might be helpful to recall the following analogous transition.
A class of languages that can be decided by a classical verifier interacting with multiple all-powerful quantum provers sharing entanglement, i.e. $\text{MIP}^*$ and of languages that can be decided by a classical verifier interacting with multiple all-powerful classical provers, i.e. MIP. A classical result from the 90s \citep{mip=nexp} shows that $\text{MIP} = \text{NEXP}$. But for a very long time it was not known if $\text{MIP} \subseteq \text{MIP}^*$. This might sounds surprising given that now we know that $\text{MIP}^* = \text{RE}$. The reason the containment $\text{MIP} \subseteq \text{MIP}^*$ was not obvious is similar to the situation in our adversarial robustness setup.
Intuitively when we allow for entanglement between provers we let honest provers to convince us of more complicated statements but dishonest provers would have more ways of deceiving the verifier. 

We borrow heavily from a series of results closely related to, the now famous, $\text{MIP}^* = \text{RE}$ paper \citep{MIPstar=RE}, more precisely the works on the delegation of quantum computation \citep{mahadev}. These techniques allow us to put a ``leash on provers with quantum capabilities''. Using ideas from these papers, we are able to design a protocol between a classical verifier and a quantum prover that guarantees that the samples collected by the verifier at the end of interaction are i.i.d. from a distribution close to $\dist$. Our main theorem provides a reduction from the problem of adversarial robustness to learning generative models and finding a classifier resistant to distributional shifts. This reduces the problem of adversarial examples to more well known problems in the machine learning field and potentially opens a way for a solution.
 

\section{Preliminaries}\label{sec:prelim}


We assume that the reader is familiar with basic notions of quantum information. For an in depth introduction we recommend \citet{nielsenbook}.

For $k \in \N$ we denote by $[k]$ the set $\{1,\dots,k\}$ and by $\mathfrak{D}(n)$ the family of distributions on $n$-bit strings. For $\mathcal{P},\mathcal{Q} \in \mathfrak{D}(n)$ we denote their Hellinger distance by $\heli(\mathcal{P},\mathcal{Q}) := \frac{1}{\sqrt{2}}\sqrt{\sum_{x \in \nbits} (\sqrt{\mathcal{P}(x)} - \sqrt{\mathcal{Q}(x)})^2} = \frac{1}{\sqrt{2}}\|\sqrt{\mathcal{P}} - \sqrt{\mathcal{Q}}\|_2$. A direct calculation yields the useful identity $1 - \heli^2(\mathcal{P}, \mathcal{Q}) = \sum_{x \in \{0,1\}^n } \sqrt{\mathcal{P}(x)\mathcal{Q}(x)}$.
We denote the total variation distance of $\mathcal{P},\mathcal{Q}$ as $\tv(\mathcal{P},\mathcal{Q}) = \frac12\|\mathcal{P}-\mathcal{Q}\|_1$. The two similarity measures satisfy 
$$\heli^2(\mathcal{P}, \mathcal{Q}) \leq \tv(\mathcal{P}, \mathcal{Q}) \leq \sqrt{2}\heli(\mathcal{P}, \mathcal{Q}).
$$

For $\dist \in \mathfrak{D}(n)$ we define $\oracle(\dist)$ as an oracle that gives access to states $ \state{\psi_{\dist}} := \sum_{x \in \nbits} \sqrt{\dist(x)} \state{x}$, where by some abuse of notation we write $\state{x}$ (for ${x \in \nbits}$) to denote $\ket{x_1}\ket{x_2}\dots\ket{x_n}$. In our protocols we will be interested in an interaction between $\ver$ (Verifier) and $\prov$ (Prover). We will write $\prov^{\oracle(\dist)}$ to denote that $\prov$ has access to $\oracle(\dist)$.
For a quantum circuit $C$ acting on $n$-qubits via the unitary transform $U_C$, we define $\mathcal{D}_C \in \mathfrak{D}(n)$ as the distribution arising from measuring all $n$ qubits of $U_C |0^{\otimes n} \rangle$ in the computational (which we will also denote as $Z$) basis. We will sometimes abuse the notation and write $C \ket{\psi}$ to mean $U_C \ket{\psi}$. For $\ket{\psi} \in (\C^2)^{\otimes n}$ we say that $\dist \in \mathfrak{D}(n)$ defined as $\dist(x) = |\braket{x}{\psi}|^2$ for every $x \in \{0,1\}^n$ is the distribution associated with $\ket{\psi}$.

Lastly we note that, often, for the sake of convenience, we drop the identity operators tensored to projection operators. For instance, instead of $I \otimes \ket{1}\bra{1} \otimes I$ we often write $\ket{1}\bra{1}$. Hence, the reader should keep in mind that there are often $I$ matrices tensored to the projection in order to make the dimensions match. 

We will use the following version of the Chernoff bound often.
\begin{fact}[Chernoff-Hoeffding]\label{fact:chernoff}
Let $X_1, \dots, X_k$ be independent Bernoulli variables with parameter $p$. Then for every $0 < \e < 1$
$$
\Prob \left[\left| \frac{1}{k}\sum_{i=1}^k X_i - p \right| > \e \right] \leq 2e^{-\frac{\e^2 k}{2}}.
$$
\end{fact}

\section{Overview}\label{sec:overview}

In this section we give an overview of the main ideas. We first explain the model of adversarial robustness we consider and then present a protocol that guarantees provable robustness in this model.

The standard (``test time'') adversarial robustness problem is defined as follows. There is a true distribution over examples $\mathcal{D}$ on an input space $\mathcal{X}$ and a ground truth function $g : \mathcal{X} \xrightarrow{} \{-1,+1\}$. During the training phase the learning algorithm generates a classifier $f = f(S) : \mathcal{X} \xrightarrow{} \{-1,+1\}$, given a sample set $S \sim_{i.i.d.} \mathcal{D}^m$. At test time the following process is repeated. A sample $\mathbf{x} \sim \mathcal{D}$ is generated and, before it is passed to $f$ for classification, an adversary $\mathbf{A}$ changes $\x$ to produce an adversarial sample $\x' = \adv(\x,f)$. In the standard setting the adversary can add to $\x$ a perturbation from a predefined ``perturbation set''. As we discussed in the introduction, there is no consensus on what perturbation sets appropriately capture real-world threat situations. We circumvent this problem and instead ask the following question
\begin{center}
    \textit{Is it possible to design an interactive protocol in which $\adv$ certifies to the learner that $\x' \sim \mathcal{D}$?}
\end{center}
In other words, can we design a procedure in which the adversary convinces the learner that the samples that are sent by $\adv$ for classification come from the proper distribution? 

\begin{wrapfigure}{r}{5.5cm}
\centering
\begin{tikzpicture}
\tikzstyle{process}=[draw,rectangle, rounded corners, fill=yellow!20,
minimum width=3cm, minimum height=1cm]
\tikzstyle{adversary}=[draw,rectangle, rounded corners, fill=red!20,
minimum width=3cm, minimum height=1cm]
\tikzstyle{nature}=[draw,rectangle, rounded corners, fill=green!20,
minimum width=3cm, minimum height=1cm]
\tikzstyle{decision}=[draw, diamond, fill=red!20, minimum width=4cm, aspect=2]
\tikzstyle{a}=[thick,->,>=stealth]
\node[nature] (nature) at (0,0) {Nature};
\node[adversary] (adversary) at (0,-3) {Adversary / Prover};
\node[process] (learner) at (0,-6) {Learner / Verifier};
\draw[a] (nature) -- node[anchor=east] {$\x \sim \mathcal{D} \text{ or } \ket{\psi_\dist}$} (adversary);
\draw[a] (adversary) -- node[anchor=east] {$\x' \text{ or } \ket{\psi_{\dist^A}}$} (learner);
\end{tikzpicture}
\caption{Model} \label{fig:model}
\end{wrapfigure}

If such a protocol existed it would guarantee that the probability of misclassification in the presence of the adversary were equal to the standard risk, i.e.,
$$\Prob_{\x \sim \mathcal{D}, \x' = \adv(\x,f)}[f(\x') \neq g(\x')] = \Prob_{\x \sim \mathcal{D}}[f(\x) \neq g(\x)].$$ Compare this to the present situation where the adversarial risk can be exponentially larger than the standard risk! 

Does such protocol solve all problems associated to the adversarial setting? Not quite! For instance, it does not guarantee any connection between $\x$ and $\x'$. Imagine for instance a self-driving car setting. A camera takes a picture $\x \sim \dist$ of a road. Assume that $\x$ represents a road that curves to the right. Next, the adversary $\adv$, who gained access to the sensor, discards $\x$ and generates a completely independent $\x' \sim \dist$. The image $\x'$ might depict a road that curves to the left. Even if $\x'$ is later on classified correctly by $f$ (as a road curving to the left) the path chosen by the car is incorrect! In such a setting $\adv$ acts as an artificial simulator of the real world. There is no way for the learner to distinguish these two cases as in both the samples come from the same distribution. Depending on the context, it might therefore be necessary to add further safeguards to ensure that the samples are connected to ``reality.'' Despite this caveat it seems that designing a protocol that certifies $\x' \sim \dist$ would be a valuable step in providing a principled solution to the adversarial robustness problem.

Is such a protocol possible even in principle? Imagine that $\dist = U(\nbits)$, i.e., $\dist$ is a uniform distribution over $n$-bit strings, and that this fact is known to both parties (in the general case neither party knows $\dist$). In this case $\adv$ should be sending uniformly random strings to the learner. But assume that $\adv$ instead prepares a string $s$ of his choosing, for instance $s = 1^n$, and tries to convince the learner that $s \sim \dist$. There seems to be no way for the learner to distinguish between an honest and a cheating $\adv$ as all strings in $\nbits$ have the same probability. If the interaction is repeated the learner can collect statistics over the runs of the interaction to try to distinguish truly random samples from malicious ones. However, usually the number of rounds of interaction is $\poly(n)$ while the size of the input space is $2^n$. Thus only an $2^{-\Omega(n)}$ fraction of the input space is explored during the interaction while $f$ usually misclassifies a constant fraction of the distribution. It hence appears that distinguishing random strings from samples that are maliciously chosen to be misclassified by $f$ is inherently hard. 

The crux of the problem seems to be the following question. How can $\adv$ convince the learner that he used true randomness when generating $s$? This is the place where quantum capabilities become useful as it is possible to answer this question in the affirmative using quantum resources! Quantum information theory tools that will help us solve this problem go back to the ideas in \citet{einsteinepr}. The implications of these ideas, phrased in modern language, are that there exist games which in order to win require the use of quantum strategies. Quantum strategies in turn imply the use of true randomness. This is the connection to our problem. 

The question of provably generating random samples from a known distribution has strong connections to a well-studied topic, namely certified randomness. In this problem a classical party, call it $\ver$ as in verifier, interacts with a quantum device (or devices), call it $\prov$ as in prover. At the end of the interaction $\ver$ should generate a string of bits that are certifiably random. This interaction should be device-independent, which means that even if $\prov$ tried to cheat in the interaction the protocol still is guaranteed to return random bits conditioned that $\ver$ accepted the interaction. This problem was studied in a variety of models. In some models it is assumed that there are several provers that share entanglement but are forbidden to communicate [\citet{Colbeck2009QuantumAR}, \citet{vaziranirandomness}, \citet{certrandtwoproversbook}]. In others there is only one prover but it has limited computational power \citep{vidickrandomness}. The key in most of these constructions is to design the protocol in such a way that the provers can only win if they use anticommutative measurements. Then the laws of quantum mechanics guarantee that the outcomes of these measurements contain true randomness.

How does the problem of certifiable randomness relate to our goal of a protocol for adversarial robustness? Recall the running example where $\dist = U(\nbits)$. Roughly spaking, a protocol for certifiable randomness between $\ver$ and $\prov$ guarantees that $\prov$ sends strings of random bits to $\ver$.
Our goal however requires us to generalize the protocol to cover all distributions and not only $\dist = U(\nbits)$. Further, as we briefly discussed earlier, neither the adversary nor the learner know $\dist$. We can therefore not simply take a black box approach but we need to start from the beginning and define the model of interaction more carefully. To follow conventions in the quantum literature, from now on we will use the term prover $\prov$ (instead of adversary $\adv$) and instead of referring to a learner we will use the term verifier $\ver$.

\paragraph{Model.}  As we discussed before, in the standard setup $\prov$ receives a sample $\x \sim \dist$, which we think was generated by nature, perturbes the sample $\x' =  \prov(\x,f)$ and sends $\x'$ to the verifier $\ver$. We refer the reader to Figure~\ref{fig:model} for a visual representation of the process. We focus on the interaction during test time. We assume that $\ver$, instead of receiving a sample $\x \sim \dist$, receives a quantum state $\ket{\psi_\dist} = \sum_{x \in \nbits} \sqrt{\dist(x)}\state{x}$, where we assumed that $\mathcal{X} = \nbits$. This model is closely related to the model that is considered in the quantum PAC-learning literature \citep{quantumPAC} where quantum samples are states of the form $\sum_{x \in \nbits} \sqrt{\dist(x)}\state{x, g(x)}$. Here, $g$ is the ground truth. Our results likely carry over to this quantum PAC-model, but as always, the details need to be verified. 

Upon receiving $\ket{\psi_\dist}$ the adversary/prover $\prov$ performs an arbitrary computation (quantum or classical) and starts an interaction with the learner/verifier $\ver$. We will consider different models depending on what type of messages can be exchanged in the protocol and what power $\ver$ has. Ultimately, we will aim for a model where $\ver$ is fully classical and the messages exchanged are classical also. However, for pedagocical reasons our presentation will proceed in three steps, where in each step we will refine the model of the verifier $\ver$. In the first step the verifier will be quantum and the messages will be allowed to contain quantum states also. In the second step, $\ver$ will have access to only a constant size quantum computer but the messages will still be allowed to be quantum. Finally, in the third step $\ver$ and all messages will be classical. As we will see, the scenarios are increasingly more complex and each builds on the previous one. This progression is also consistent with the progression of models in the research literature on certified randomness and delegation of quantum computation - fields which we heavily borrow from.

\paragraph{Protocol. } Observe that if you measure $\ket{\psi_\dist} = \sum_{x \in \nbits} \sqrt{\dist(x)}\state{x}$ in the computational basis then the outcome of the measurement is distributed according to $\dist$. Hence, if we can design a protocol that forces $\prov$ to measure $\ket{\psi_\dist}$ in the $Z$ basis and send the result to $\ver$ then this will fulfill our goal.

In particular, we want the protocol to satisfy two properties: completeness and soundness. The completeness property guarantees that if $\prov$ acts honestly then the interaction will be accepted by $\ver$ with high probability and the outcome of the protocol will be ``equivalent'' to $\prov$ measuring $\ket{\psi_\dist}$ in the $Z$ basis and sending the results to $\ver$. The soundness property guarantees that ``the only'' way to succeed in the protocol is to obtain a state $\ket{\psi_\dist}$, measure it in the computational basis, and send the result to $\ver$. 

We will now sketch our protocols and their guarantees in the three models.

\paragraph{Quantum Verifier. } The key component of all our protocols is a quantum circuit $\comp$, acting on three registers: out ($1$ qubit), adv ($n$ qubits) and aux ($n$ qubits), depicted in Figure~\ref{fig:CS1}. $\comp$ is parametrized by a quantum circuit $C$ with the associated distribution $\dist_C$. Recall that the result of applying $U_C$ to $0^{\otimes n}$ is the state $\ket{\psi_{\dist_C}}$. The circuit is designed so that it measures the similarity between $\dist^A$ and $\dist_C$, where $\dist^A$ is the distribution corresponding to the state $\ket{\psi_{\dist^A}}_\adver$. More precisely, the closer $\dist^A$ and $\dist_C$ are in terms of the Hellinger distance the higher the probability that $\comp$ outputs $1$ in the out register.

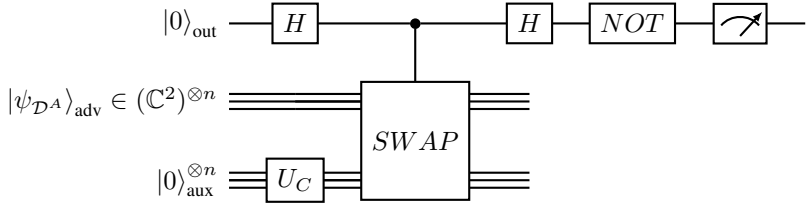
\begin{figure}
\centering
\begin{quantikz}
\lstick{$\ket{0}_{\text{out}}$} & \gate{H}  & \ctrl{1} & \gate{H} & \gate{NOT} & \meter{} & \qw  \\
\lstick{$\ket{\psi_{\dist^A}}_{\text{adv}} \in (\mathbb{C}^2)^{\otimes n }$} &  \qwbundle[alternate]{} & \gate[wires=2]{SWAP} \qwbundle[alternate]{} & \qwbundle[alternate]{} &&& \\
\lstick{$\ket{0}^{\otimes n}_{\text{aux}}$} & \gate{U_C} \qwbundle[alternate]{} & \qwbundle[alternate]{} & \qwbundle[alternate]{} &&&
\end{quantikz}
\caption{Comparison Circuit} \label{fig:CS1}
\end{figure}

Equipped with such a comparison circuit we are ready to design a protocol in a model where $\ver$ has quantum capabilities. For now we assume that $\prov$ acts i.i.d. in every round of the protocol. We will come back and discuss this assumption at the end of this overview. In the $i$-th round of the interaction $\prov$ sends an $n$-qubit quantum state $\ket{\psi_{\dist^A}}$ to $\ver$, $\ver$ samples a bit $b_i \in \{0,1\}$ uniformly at random. If $b_i=0$ then $\ver$ inserts $\ket{\psi_{\dist^A}}$ as an input to $\comp$, computes $\comp$, measures the output bit in the $Z$ basis and records the result as $\gamma_i$. If $b=1$ then $\ver$ measures $\ket{\psi_{\dist^A}}$ in the $Z$ basis and records the outcome as $\x_i \in \nbits$. After a certain number of rounds (dependent  on the desired accuracy and probability of success) $\ver$ computes an average $\gamma_{\text{avg}}$ of the set $\{\gamma_i: b_i=0\}$. If $\gamma_{\text{avg}}$ is bigger than a certain (do be determined) threshold $\ver$ accepts the interaction and returns the set $\{\x_i: b_i=1\}$.

Let us now consider the properties of this protocol. Completeness of the protocol is straightforward. An honest $\prov$ can forward the state $\ket{\psi_\dist}$ he receives to $\ver$ (see Figure~\ref{fig:model} for the visualization of the model). For soundness of the protocol note the following facts: i) $\gamma_{\text{avg}}$ is a good approximation for the probability that $\comp$ outputs $1$ on $\ket{\psi_{\dist^A}}$, ii) this probability is monotonically related to $\heli(\dist^A,\dist_C)$ by the properties of $\comp$, iii) the samples $\{\x_i: b_i=1\}$ are 
 i.i.d. from $\dist^A$, iv) we assumed that $d_H(\dist, \dist_C) $ is small. Moreover we assume that $\heli(\dist, \dist_C) \approx \eta$ and that $\eta$ is known to $\ver$. Combining these facts we arrive at the following conclusion. If $\ver$ accepts the interaction then the samples it returns are i.i.d. from a distribution $\dist^A$ such that 
\begin{equation}\label{eq:mainguarantee}
\heli(\dist^A, \dist) < O(\heli(\dist_C, \dist)).
\end{equation}

The reason the above holds is because we can set the threshold in the protocol over which $\ver$ accepts $\gamma_{\text{avg}}$ to be such that the interaction is accepted when $\heli(\dist^A, \dist_C) \lesssim \eta$. Then using a triangle-like inequality we arrive at \eqref{eq:mainguarantee}. We can summarize this guarantee by the following, informal theorem.

\begin{theorem}[Informal]\label{thm:informal}
For every circuit $C$ there exists a protocol between a verifier $\ver$ and a prover $\prov$ such that for every $\dist \in \mathfrak{D}(n)$ the protocol guarantees that, conditioned on $\ver$ accepting the interaction, the samples returned by $\ver$ are i.i.d. from a distribution $\dist^A \in \mathfrak{D}(n)$ such that
$$
\heli(\dist^A, \dist) <  O(\heli(\dist_C, \dist)).
$$
\end{theorem}

We can rephrase this result by saying

\begin{center}
\textbf{Adversarial robustness in the quantum world is no harder than \\ learning generative models + robust wrt distributional shifts.} 
\end{center}
Let us explain this rephrasing in more detail.
Imagine that we were able to learn a generative model $C$ that captures the true distribution well, i.e., $\heli(\dist_C, \dist) = \eta \ll 1$. Imagine further that we were also able to learn a classifier that is robust with respect to small changes in the distribution (i.e., it is robust to distributional shifts). This means that for all $\dist^A \in \mathfrak{D}(n)$ such that $\heli(\dist^A, \dist) \leq O(\eta)$ we would have $R_{\dist^A}(f) \approx R_{\dist}(f)$ (where $R$ denotes the standard risk). Then using the protocol guaranteed by Theorem~\ref{thm:informal} we would solve the adversarial robustness problem in our quantum model.

Is it realistic to assume that we can find a generative model such that $d_H(\dist_C, \dist) \ll 1$ and a classifier robust wrt distributional shifts? 

First we address the assumption about learning generative models. Up until now we focused on the learning/inference pipeline during test time. The circuit $C$ however, might, and usually will, depend on the learning phase. As we discussed before, neither $\prov$ nor $\ver$ know $\dist$ at the beginning of the learning process.\footnote{$\ver$ and $\prov$ could have some previous knowledge of $\dist$ before learning starts but it would not affect our results.} We assume however that throughout the learning phase $\ver$ ``learns'' an approximation to $\dist$. More precisely, as is standard for classical machine learning, during the learning phase $\ver$ receives classical samples from $\dist$ and based on them computes a description of a circuit $C$ (that might be a classical\footnote{A classical circuit that generates samples from a distribution uses external randomness. We can realize such a circuit by a quantum circuit by appending it with an ancilla bit initialized to $\ket{+}$ for every bit of randomness used.} or a quantum circuit) such that $d_H(\dist ,\dist_C) \ll 1$. We note that we could have also considered a model where $\ver$ receives quantum samples $\left(\sum_{x \in \nbits} \sqrt{\dist(x)}\state{x} \right)$ during learning. Since  in this paper we are interested mostly in the test phase, we will not focus on this distinction further. Let us emphasize that the assumption about finding good generative models is natural in many situations. E.g., generative neural networks give rise to circuits with associated distributions that approximate the true distribution.\footnote{For the case of generative neural networks it is enough to represent a neural network as a circuit, which is possible as one can implement basic arithmetic operations with a standard set of gates.} 

The task of learning generative models is of course different from the standard supervised learning task, where the goal is to learn a good predictor from labeled samples. In general these two tasks are not comparable. It is reasonable to assume that in some settings learning a generative model is harder (in terms of sample complexity) than learning a predictor. This means that, as per Theorem~\ref{thm:informal}, more samples are needed to run our protocol than to run the standard supervised learning. But some increase in the sample complexity might be unavoidable as learning in the presence of an adversary is a harder problem than the standard supervised learning. 

What about the assumption that we are able to learn a classifier that is robust wrt to distributional shifts? At first glance this looks similar to the original adversarial robustness problem, namely constructing classifiers that are robust to a given perturbation set, but now the perturbation set is replaced by a constraint on the Hellinger distance between the original and the resulting distribution. Have we made progress with our protocol? We claim that yes! Recall the standard setup where the perturbations of the adversary $\adv$ are bounded in the $\ell_2$ norm by $\e$. Moreover assume that we have a classifier $f : \inspace \rightarrow \outspace$ with an error set $E$ and a standard risk $R_\dist(f) = \Prob_{x \sim \dist}[x \in E] \ll 1$. In general $\Prob_{x \sim \dist}[x \in E] \ll \Prob_{x \sim \dist}[x \in E + B_\e] $, where $B_\e$ denotes the $\ell_2$ ball of radius $\e$ and $+$ denotes the Minkowski sum. In fact, the right-hand side can be of order $1$.
We can look at the actions of $\adv$ as generating a distribution $\dist'$. More specifically, let $p : \inspace \rightarrow \inspace$ be the perturbation function that $\adv$ applies, e.g. $p(E + B_\e) \subseteq E$. We define $\dist'$ as the result of the following process (i) $x \sim \dist$ (ii) $x' = p(x)$. In that language we can rephrase the adversarial robustness problem (under $\ell_2$ perturbations) as finding classifiers resistant to distributional shifts of this form. But observe that $d_H(\dist, \dist')$ is high in this case as a considerable amount of mass (at least $\Prob_{x \sim \dist}[x \in E + B_\e] - \Prob_{x \sim \dist}[x \in E]$) was moved. This shows that the class of distributions being close to $\dist$ in the Hellinger distance is smaller than those considered in the standard adversarial robustness setups. This makes the task of learning a classifier robust to distributional shifts (under the Hellinger distance) easier. We note that the robustsness to distributional shifts assumption is also natural in some situations. There is a considerable literature about finding such classifiers, see for instance \citet{shiftsbook}, \citet{carlinishifts}, \citet{imagenettoimagenet} or the PhD thesis of Martin Arjovsky for a survey \citep{arjovskythesis}.
 
The first key contribution of our result is that it reduces adversarial robustness to two well-studied problems in the field of machine learning, namely
the problem of finding generative models and the problem of devising classifiers that are robust wrt distributional shifts.
Of course, these two problems might be harder than the standard supervised learning setup. But reductions of this sort embed the problem of adversarial robustness firmly into the learning theory framework. The second key contribution is that in our framework does not rely on specific (and hence also somewhat arbitrary) threat models like bounded perturbations. Instead we identify a single distance notion (the Hellinger distance) that suffices to consider to solve the adversarial robustness problem in our model. 


This protocol might not be fully satisfactory as we require the verifier to have quantum capabilities. Our aim is to prove a similar result to Theorem~\ref{thm:informal} in the case where the verifier is fully classical. As a first step towards this goal we describe a protocol for a model where $\ver$ has access to a constant-size quantum computer. In the final step we then generalize the result to a model where $\ver$ is classical.

\begin{figure}
\centering
\begin{quantikz}
\lstick{$\ket{0}_{\text{out}}$} \slice{$t = 0$} & \qw \slice{$T$} & \gate{H} \slice{$T+1$}  &  \ctrl{1} \slice{$T+n+1$} & \gate{H} \slice{} & \gate{NOT} \slice{$T+n+3$} & \meter{} & \qw  \\
\lstick{$\ket{\psi_\dist}_{\text{adv}} \in (\mathbb{C}^2)^{\otimes n }$} & \qwbundle[alternate]{} & \qwbundle[alternate]{} & \gate[wires=2]{SWAP} \qwbundle[alternate]{} & \qwbundle[alternate]{} & \qwbundle[alternate]{} & \qwbundle[alternate]{} & \qwbundle[alternate]{} \\
\lstick{$\ket{0}^{\otimes n}_{\text{aux}}$} & \gate{U_C} \qwbundle[alternate]{} & \qwbundle[alternate]{} & \qwbundle[alternate]{} &  \qwbundle[alternate]{} & \qwbundle[alternate]{} & \qwbundle[alternate]{}  & \qwbundle[alternate]{}
\end{quantikz}
\caption{Comparison Circuit with Time Slices} \label{fig:CS2}
\end{figure}
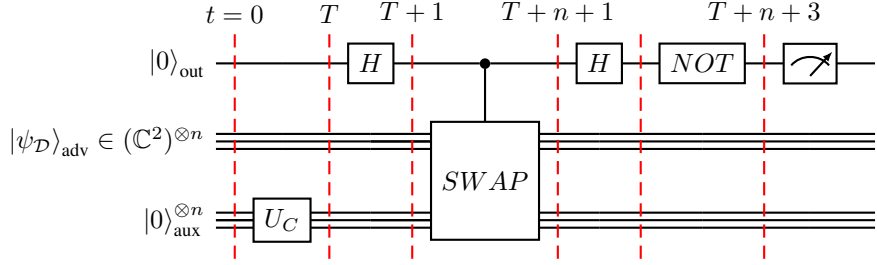

\paragraph{Constant Memory Quantum Verifier. } In this model the messages in the protocol can still be quantum. (But we will see that in our protocol only $\prov$ will send quantum states and $\ver$ will send only classical messages.) But $\ver$ now only has access to a constant-size quantum computer and can store only a constant number of qubits at each point in time $\ver$. The only operation that will be required from $\ver$ is measuring the qubits sent by $\prov$ in either the $Z$ or the $X$ basis. Protocols of this form are called receive-and-measure protocols and were already previously considered in the literature, see e.g. \cite{fitzsimons}.

Our goal is to emulate the protocol that we designed in the previous step in this more restrictive constant-quantum memory model. The idea is the following. We let $\prov$ choose an $n$-qubit state $\ket{\psi_{\dist^A}}$ and then force her to create a state $\ket{\phi}$ that depends on $\ket{\psi_{\dist^A}}$ and to send this state to $\ver$.\footnote{By sending the state to $\ver$ we mean sending the state one qubit at a time. Whenever a qubit arrives to $\ver$ he has a choice whether to keep it or discard it. At all times the number of qubits $\ver$ stores cannot exceed the constant predefined number.} The state $\ket{\phi}$ should satisfy the following properties. When $\ver$ measures $\ket{\phi}$ in the $Z$ basis then (i) with probability $\Omega(1/T)$ the distribution of outcomes of measuring one of the qubits is close to the distribution of measuring the output qubit of $\comp \ket{0}\ket{\psi_{\dist^A}}\ket{0^{\otimes n}}$ in the $Z$ basis (ii) with probability $\Omega(1/T)$ $\ver$ can obtain $\x' \sim \dist^A$.\footnote{Note that $\x' \in \nbits$ but $\ver$ has only a constant quantum memory. But it is possible to realize a protocol with these properties. Imagine that while the qubits come to $\ver$ one by one he measures a qubit, records the result and discards the qubit making room for the next ones. In total he collects many measurement outcomes out of which he can create $\x'$.} These two operations emulate the steps $\ver$ performed in the previous protocol for $b = 0$ and $b  = 1$, respectively. Note that the operations succeed only with probability $\Omega(1/T)$ but this suffice for our purpose. The main question is how to force $\prov$ to create $\ket{\phi}$ with these properties?

To solve this problem we use the well-known circuit-to-Hamiltonian reduction introduced in \citet{kitaevbook}. This reduction was originally used to show that a local Hamiltonian problem is QMA-complete. Later on it was a crucial component in the the delegation of quantum computation in the constant quantum memory model \citep{fitzsimons} and in the delegation of quantum computation with a classical verifier in \citet{mahadev}. The properties that we require from the reduction are stronger than the ones usually needed. Due to these stringent constraints we cannot use the reduction in a black-box form but need to analyze it in detail.

What is the purpose of this reduction in our context? The circuit-to-Hamiltonian reduction allows to reduce the computation of a quantum circuit $\comp$ to estimating an energy of a state $\ket{\rho}$ with respect to a local Hamiltonian $H_\comp$. In particular, it allows us to build a protocol that forces $\prov$ to prepare a so-called \textit{history state} $\ket{\phi}$ of $\comp$. Assume that $\prov$ chooses to evaluate $\comp$ on a state $\ket{\psi_{\dist^A}}$. Assume further that the circuit $C$ has $T$ gates and denote $T+n+3$ by $T'$. Then denote by $\ket{\xi_0}, \ket{\xi_1}, \dots, \ket{\xi_{T+n+3}}$ the $(2n+1)$-qubit states, where $\ket{\xi_i}$ is the state after $\comp$ is computed for $i$ steps on $\ket{\psi_{\dist^A}}$. We refer the reader to Figure~\ref{fig:CS2}, where the $\xi_i$'s are depicted as time slices in $\comp$. With this notation the history state is defined as:
$$
\ket{\phi} := \frac{1}{\sqrt{T'+1}} \left( \ket{0}_{\text{clock}}\ket{\xi_0}_{\text{comp}} + \ket{1}_{\text{clock}}\ket{\xi_1}_{\text{comp}} + \ket{2}_{\text{clock}}\ket{\xi_2}_{\text{comp}} + \dots + \ket{T'}_{\text{clock}}\ket{\xi_{T'}}_{\text{comp}} \right).
$$
Hence, $\ket{\phi}$ represents a history of the evaluation of $\comp$. It is a superposition of states of the circuit after applying $0,1,2 \dots, T'$ gates of the circuit tensored with a state representing a \textit{clock}. We denoted by comp the concatenation of the three registers out, adv, aux. For instance $\ket{\xi_0}_{\text{comp}} = \ket{0}_{\text{out}}\ket{\psi_{\dist^A}}_{\text{adv}}\ket{0^{\otimes n}}_{\text{aux}}$.

Assume for now that $\prov$ sends $\ket{\phi}$ to $\ver$. We will show that with such a state it is possible to realize the two properties we were hoping for. $\ver$ measures $\ket{\eta}$ in the $Z$ basis and depending on the outcome of measuring the clock register performs further actions.

If the outcome of measuring the clock register is equal to $T'$, which by definition of $\ket{\phi}$ happens with probability $\frac{1}{T'+1}$, then the distribution of measuring the out register is exactly equal to the desired distribution. This is because $\ket{\xi_{T'}}_{\text{comp}}$ represents the last slice of the computation of $\comp$ (see Figure~\ref{fig:CS2}). 

If the outcome of measuring the clock register is equal to $0$ then the distribution of measuring the adv register is exactly equal to $\dist^A$. This is because $\ket{\xi_{0}}_{\text{comp}} = \ket{0}\ket{\psi_{\dist^A}}\ket{0^{\otimes n}}$ represents the first slice of the computation of $\comp$ (see Figure~\ref{fig:CS2}). Note that in the final protocol we also check if the outcomes of measuring the out and aux registers are all $0$. This is done for technical reasons to simplify the proof of soundness.

We realized the two properties we were looking for. Now we can emulate the protocol described in the first step (Quantum Verifier). Thus we will obtain a result similar to Theorem~\ref{thm:informal} also in this setting. Note that in each of the cases we were succeeding only with probability $\approx 1/T'$. This will influence the guarantee of Theorem~\ref{thm:informal} in this model. In particular, this will imply that we will recover $1$ sample from $\dist$ for every $T'$ states $\ket{\psi_\dist}$ provided to an honest $\prov$.

In Section~\ref{sec4.2} we will explain in more detail what it formally means that we can force $\prov$ to produce the history state. In short, the circuit-to-Hamiltonian reduction allows $\ver$ to perform local (which means involving only few qubits) checks on the state obtained from $\prov$ to check that it is in fact a history state. These local checks and the whole reduction has a flavor similar to the famous Cook-Levin proof that shows that 3-SAT is NP-complete.

\paragraph{Classical Verifier.} In the last model we consider $\ver$ that is classical and all exchanged messages are also classical. To make our protocol work we need to impose a computational restriction on $\prov$, namely we assume that $\prov$ is in QPT- Quantum Polynomial Time. 

The goal now is to adopt the protocol from the previous step to this model. The protocol can be understood as forcing $\prov$ to construct a history state by performing checks (measurements in the $X$ or the $Z$ basis) that involve only constant number of qubits. In the model where the communication is only classical we need to somehow force $\prov$ to perform the measurements chosen by $\ver$ and report the result of these measurements back to $\ver$. 

To achieve this we use an idea that was a crucial component in the delegation of quantum computation with a classical verifier in \citet{mahadev}. A similar idea was used previously in \citet{vidickrandomness} to generate certified randomness with a classical verifier. On a high level, we design a protocol that forces $\prov$ to commit to an $n$-qubit state $\ket{\phi}$, then receive instructions for measurements from $\ver$, measure $\ket{\phi}$ accordingly and report the results back to $\ver$. This procedure relies on the concept of claw-free functions, which can be understood as a post quantum cryptography scheme that allows $\ver$ to control the computation performed by $\prov$.


We will not go into more detail in this overview and we refer the reader to Section~\ref{sec4.3} for a comprehensive treatment. We summarize the overview by saying that also in this model we arrive at a result similar to Theorem~\ref{thm:informal}. This means that if we assume that we are able to learn generative models close to the true distribution in the Hellinger distance and our classifiers are resistant to small distributional shifts in the same distance measure then adversarial robustness is solved in a model where a classical $\ver$ interacts with a quantum $\prov$ who receives quantum samples from nature.

\paragraph{Discussion. } For the most part of the paper we assume that the $\prov$ acts in an i.i.d. fashion. For the fully-quantum verifier we give a proof also for general setting where we drop the i.i.d. assumption. To keep the exposition manageable we do not provide general proofs for the other two models. 

How can we check the success probability of the prover?
Assuming that the prover behaves in an i.i.d. fashion, it suffices to run the protocol $(2/\epsilon) \log(1/\delta)$ times. If the fraction of successes is bigger than $1-\e/2$ then we know with confidence $ 1 -\delta$ that the success probability is at least $1-\e$.



\section{Adversarial Robustness Protocols}

We are now ready to define our protocols formally and analyze their correctness. We start with a protocol in a model where $\ver$ has access to a quantum computer, then move on to a model where $\ver$ only has a constant memory quantum computer, and we conclude with a model where $\ver$ is fully classical. Each of these models is presented in its own section. 

\subsection{Quantum Verifier}\label{sec4.1}

As we described in Section~\ref{sec:overview} all protocols will be based on the Comparison Circuit defined in Figure~\ref{fig:CS1}. We note that a circuit of this form, i.e., a circuit where we first apply the H gate to a control qubit, then apply a circuit conditioned on that qubit, and finally apply H to the qubit again, is a key component of many quantum algorithms \citep{quantumalgbook}. It can be viewed as a low-precision
form of the eigenvalue estimation procedure associated with Shor’s algorithms for
factoring and computing discrete logarithms. 

Recall that the high-level idea for the protocol, in the model where $\ver$ has access to a quantum computer, is as follows. $\ver$ receives $\ket{\psi_{\dist^A}}$ from $\prov$, $\ver$ flips a coin and depending on the outcome either measures $\ket{\psi_{\dist^A}}$ in the $Z$ basis, to generate a sample, or evaluates $\comp$ on $\ket{\psi_{\dist^A}}$ to estimate $\heli(\dist^A, \dist_C)$. If the estimate for the distance is sufficiently small, $\ver$ accepts the interaction and returns the collected samples. The protocol is defined in Figure~\ref{fig:quantumverifier_protocol}. We start by assuming that $\prov$ acts in an i.i.d. fashion and that the states $\prov$ sends are pure. We discuss how to remove these assumptions in Sections~\ref{sec:quantumnoniid} and \ref{sec:quantumvermixedstates}.

\begin{figure}
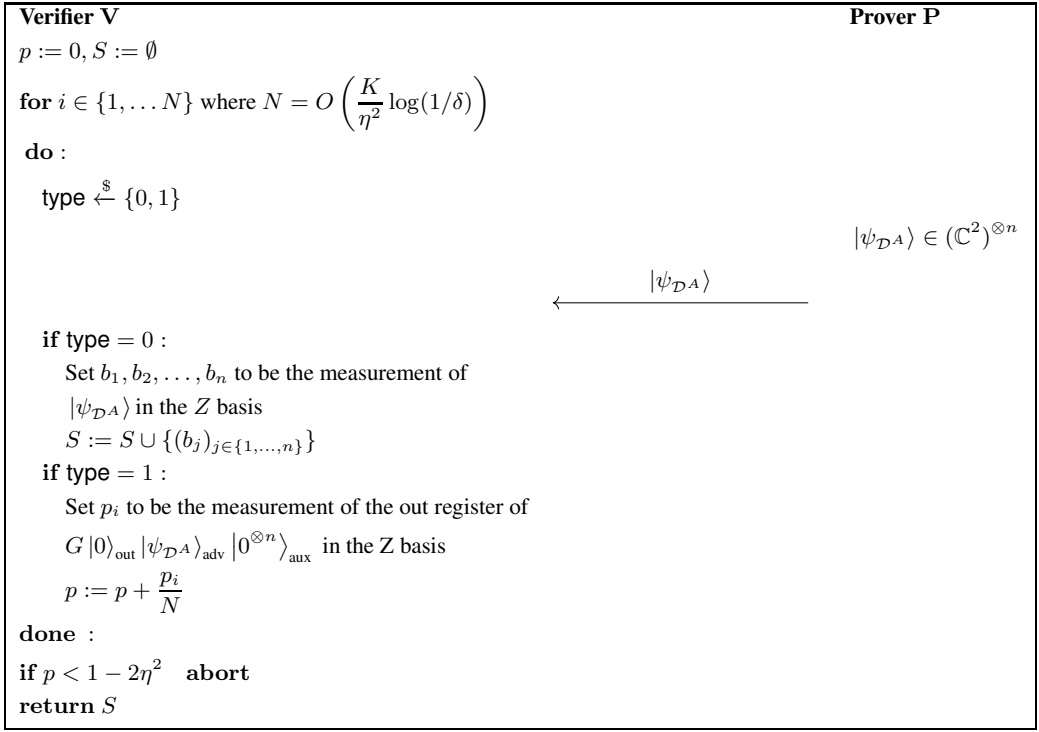

	\scalebox{0.97}{
        \fbox{
            \procedure{}{
                \textbf{Verifier } \ver \>\>  \textbf{Prover } \prov	\\
                p := 0, S := \emptyset \> \> \\
                \pcfor i \in \{1,\dots N\} \text{ where } N = O \left(\frac{K}{\eta^2} \log(1/\delta)\right) \> \>\\ 
                \pcdo:\>\> \\
                \t \textsf{type} \xleftarrow[]{\$} \{0,1\}\> \> \\
                \>\> \ket{\psi_{\dist^A}} \in
                (\mathbb{C}^2)^{\otimes n} \\
                \> \sendmessageleft*{\ket{\psi_{\dist^A}}} \> \\ 
                \t \pcif \textsf{type} = 0: \\ 
                \t \t \text{Set } b_1, b_2, \dots, b_n \text{ to be the measurement of } \\
                \t \t \ket{\psi_{\dist^A}} \text{in the $Z$ basis}\\
                \t \t S := S \cup \{(b_j)_{j\in\{1 ,\dots , n\}}\}\\
                \t \pcif \textsf{type} = 1: \\ 
                \t \t \text{Set } p_i \text{ to be the measurement of the out register of } \\
                \t \t \comp \ket{0}_\out \ket{\psi_{\dist^A}}_\adver \ket{0^{\otimes n}}_\aux \text{ in the Z basis} \\
                \t \t p := p + \frac{p_i}{N}\\
                \pcdone : \\
                \pcif p < 1 -2\eta^2 
                \t \pcabort \\
                \pcreturn S
                }
        }
    }
    \caption{The interactive protocol for the model where the verifier has access to a quantum computer.}
    \label{fig:quantumverifier_protocol}
\end{figure}

Let us prove the correctness of this protocol. We start by showing how the distribution of measuring the out register of $\comp$ relates to the Hellinger distance of $\dist_C$ and $\dist^A$.

\begin{lemma}\label{lem:similaritymeasure}
For every $\ket{\psi_{\dist^A}}$ and $C$ the probability of obtaining outcome $\ket{1}$ when measuring the out register of $\comp \ket{0}_\out\ket{\psi_{\dist^A}}_\adver \ket{0^{\otimes n}}_\aux$  is equal to
$$ 
\frac12 \left(1 + \left|\braket{\psi_{\dist^A}}{\psi_{\dist_C} } \right|^2 \right).
$$
\end{lemma}
\begin{proof}
We analyze the evolution of the state
\begin{align*}
&(\text{NOT} \otimes \id \otimes \id)(H \otimes \id \otimes \id)(\text{CSWAP})(H \otimes \id \otimes \id)(\id \otimes \id \otimes U_C)\ket{0}\ket{\psi_{\dist^A}}\ket{0^{\otimes n}}
 \\
& = (\text{NOT} \otimes \id \otimes \id)(H \otimes \id \otimes \id)(\text{CSWAP})(H \otimes \id \otimes \id)\ket{0}\ket{\psi_{\dist^A}}\ket{\psi_{\dist_C}}
 \\
& = (\text{NOT} \otimes \id \otimes \id)(H \otimes \id \otimes \id)(\text{CSWAP})\left(\frac{\ket{0} + \ket{1}}{\sqrt{2}} \right)\ket{\psi_{\dist^A}}\ket{\psi_{\dist_C}} \\
& = (\text{NOT} \otimes \id \otimes \id)(H \otimes \id \otimes \id) \frac{1}{\sqrt{2}}\left(\ket{0}\ket{\psi_{\dist^A}}\ket{\psi_{\dist_C}} + \ket{1}\ket{\psi_{\dist_C}}\ket{\psi_{\dist^A}} \right)\\
& = (\text{NOT} \otimes \id \otimes \id)\frac{1}{2}\left((\ket{0} + \ket{1})\ket{\psi_{\dist^A}}\ket{\psi_{\dist_C}} + (\ket{0} - \ket{1})\ket{\psi_{\dist_C}}\ket{\psi_{\dist^A}} \right)\\
& = \frac{1}{2}\left(\ket{1} \left[\ket{\psi_{\dist^A}}\ket{\psi_{\dist_C}} + \ket{\psi_{\dist_C}}\ket{\psi_{\dist^A}} \right] + \ket{0}\left[\ket{\psi_{\dist^A}}\ket{\psi_{\dist_C}} -\ket{\psi_{\dist_C}}\ket{\psi_{\dist^A}} \right] \right).
\end{align*}
The probability of obtaining outcome $\ket{1}$ when measuring the out register in the $Z$ basis is then
\begin{align*}
&\frac14 \left[ \left(\bra{\psi_{\dist^A}}\bra{\psi_{\dist_C}} + \bra{\psi_{\dist_C}}\bra{\psi_{\dist^A}}\right) \left(\ket{\psi_{\dist^A}}\ket{\psi_{\dist_C}} + \ket{\psi_{\dist_C}}\ket{\psi_{\dist^A}}\right) \right] \\
& = \frac14 [ \braket{\psi_{\dist^A}}{\psi_{\dist^A}}\braket{\psi_{\dist_C}}{\psi_{\dist_C}} + \braket{\psi_{\dist^A}}{\psi_{\dist_C}}\braket{\psi_{\dist_C}}{\psi_{\dist^A}} + \\
&+\braket{\psi_{\dist_C}}{\psi_{\dist^A}}\braket{\psi_{\dist^A}}{\psi_{\dist_C}} + \braket{\psi_{\dist_C}}{\psi_{\dist_C}}\braket{\psi_{\dist^A}}{\psi_{\dist^A}} ] \\
&=\frac14 \left[ 2 \|\psi_{\dist^A}\|^2 \|\psi_{\dist_C}\|^2 + 2 \braket{\psi_{\dist_C}}{\psi_{\dist^A}}\braket{\psi_{\dist^A}}{\psi_{\dist_C}} \right] \\
&=\frac12 \left[ 1 + |\braket{\psi_{\dist^A}}{\psi_{\dist_C}}|^2 \right].
\end{align*}
\end{proof}

\begin{corollary}\label{cor:prelationtohellinger}
The probability of obtaining outcome $\ket{1}$ when measuring the out register of $\compcirc$ executed on $\ket{\psi_{\dist^A}}$, which also can be expressed as $\bra{0^{\otimes n}}_\aux \bra{\psi_{\dist^A}}_\adver \bra{0}_\out G^\dagger \Pi_{\out}^{(1)} \compcirc \ket{0}_\out \ket{\psi_{\dist^A}}_\adver \ket{0^{\otimes n}}_\aux$ is equal to
$$
\frac12 \left(1 + (1 - \heli ^2(\dist^A, \dist_C))^2  \right).
$$
\end{corollary}

\begin{proof}
Apply Lemma~\ref{lem:similaritymeasure} and the $1 - \heli^2(\mathcal{P}, \mathcal{Q}) = \sum_{x \in \{0,1\}^n } \sqrt{\mathcal{P}(x)\mathcal{Q}(x)}$ identity.
\end{proof}


Next we show that the number of times each of the types ($0$ and $1$) occurs is at least $N/4$ with high probability. This is a simple application of the Chernoff bound.

\begin{lemma}\label{lem:eventforquantumverifier}
    Let $n_0,n_1$ be the number of times each type occurs in the protocol from Figure~\ref{fig:quantumverifier_protocol}. If $N = \Omega(\log(1/\delta))$ then $\Prob[n_0,n_1 > \frac{N}{4}] \geq 1 - \delta$.
\end{lemma}
\begin{proof}
For $b \in \{0,1\}$, $n_b$ can be seen as a sum of random Bernoulli variables $\{x_i\}_{i\in [N]}$ with parameter $1/2$. Then, by Fact~\ref{fact:chernoff}, we get that $\Prob[|\frac{n_b}{N} - \frac12| > \frac14] \leq 2e^{-\frac{N}{32}} \leq \frac\delta2$. We finish by applying the union bound over the error events.
\end{proof}






We are now ready to combine all the pieces and prove that the protocol from Figure~\ref{fig:quantumverifier_protocol} guarantees that if $\ver$ accepts the interaction then with high probability the samples he collected are i.i.d. from a distribution close to $\dist$.

\begin{theorem}[Quantum Verifier]\label{thm:quantumverifier}
For every circuit $C$ acting on $n$ qubits, for every $\delta \in (0,\frac13), K \in \N$ and all $\eta > 0$ sufficiently small there exists an interactive protocol between a quantum verifier $\ver$ and a quantum prover $\prov$ with the following properties. The protocol runs in $N = O\left(\frac{K}{\eta^2} \log(1/\delta)\right)$ rounds and in each round $\prov$ sends a pure quantum state on $n$ qubits to $\ver$. At the end of the protocol $\ver$ outputs $\bot$ when it rejects the interaction or it outputs $ S = \{x_1, \dots, x_{|S|} \}$, $x_i \in \nbits$, when it accepts.

\begin{itemize}
    \item (Completeness)There exists $\prov^{\oracle(*)}$ such that for every $\mathcal{D} \in \mathfrak{D}(n)$ satisfying $\heli(\mathcal{D}, \mathcal{D}_C) \leq \eta$ the following holds. With probability $1- \delta$ over the randomness in the protocol  $\prov^{\oracle(\mathcal{D})}$ succeeds, $S \sim_{\text{i.i.d.}} \mathcal{D}^{|S|}$, and $|S| \geq \Omega(K)$.
    \item (Soundness) For every $\prov$ that succeeds with probability at least $ \frac23$ we have $S \sim_{\text{i.i.d.}} (\dist^A)^{|S|}$ and $\heli(\dist_C, \dist^A) \leq O(\eta)$. \footnote{$\dist^A$ is the implicit distribution from which we collect the samples, which is the distribution corresponding to $\ket{\psi_{\dist_A}}$} 
\end{itemize}
\end{theorem}

\begin{proof}
We start with the completeness property and then move to soundness.

\paragraph{Completness.} An honest $\prov^{\oracle(\dist)}$ obtains $\ket{\psi_\dist}$ from $\oracle(\dist)$ and forwards it to $\ver$. Lemma~\ref{lem:eventforquantumverifier} guarantees that with probability $1- \frac{\delta}{2}$, $n_0,n_1 = \Omega(\frac{K}{\eta^2} \log(1/\delta))$. This automatically guarantees that $|S| \geq \Omega(K)$. 
Moreover by Fact~\ref{fact:chernoff} we have that with probability $1-\frac{\delta}{2}$ 
$$
\left| p - \bra{0^{n}}_\aux \bra{\psi_{\dist}}_\adver \bra{0}_\out G^\dagger \Pi_{\out}^{(1)} \compcirc \ket{0}_\out \ket{\psi_{\dist}}_\adver \ket{0^{n}}_\aux \right| \leq \eta^2.
$$

By Corollary~\ref{cor:prelationtohellinger} we thus get that
\begin{equation}\label{eq:estimateofp}
\left|p - \frac12 \left(1 + (1 - \heli ^2(\dist, \dist_C))^2  \right) \right| \leq \eta^2,
\end{equation}
holds with probability $1 - \frac\delta2$. By assumption $\heli(\dist, \dist_C) \leq \eta$ so we get that $p \geq  1 - 2\eta^2$ (as a function $\frac{1}{2}(1+(1-x^2)^2)$ is decreasing). This means that $\prov^{\oracle(\mathcal{D})}$ succeeds with probability $1 -\delta/2$. 

By the union bound over the error events with probability $1 - \delta/2 - \delta/2 = 1 - \delta$ we have that $|S| \geq \Omega(K)$ and $\prov^{\oracle(\dist)}$ succeeds. The property $S \sim_{\text{i.i.d.}} \mathcal{D}^{|S|}$ holds because the state sent by $\prov$ to $\ver$ is equal to $\ket{\psi_\dist}$.
\paragraph{Soundness.} By Corollary~\ref{cor:prelationtohellinger} we get that with probability $1 - \frac\delta2$ 
$$\left|p - \frac12 \left(1 + (1 - \heli ^2(\dist^A, \dist_C))^2  \right) \right| \leq \eta^2.
$$
$\prov$ succeeds with probability $\frac23$ so by the union bound and the fact that $\frac13 + \frac\delta2 < 1$ we get that $f(\heli(\dist^A, \dist_C)) \geq p - \eta^2 \geq 1 - 3\eta^2$, where we used $f$ to denote the function $\frac{1}{2}(1 + (1-x^2)^2)$. As $f$ is a decreasing function we get that that $\heli(\dist^A, \dist_C) \leq \sqrt{1 - \sqrt{2(1-3\eta^2) - 1}} \leq 10\eta$, for sufficiently small $\eta$.
\end{proof}

\subsubsection{Non i.i.d. Quantum Verifier}\label{sec:quantumnoniid}
Let us now relax the assumption that $\prov$ acts i.i.d., i.e. that $\prov$ sends the same $\ket{\psi_{\dist^A}}$ in every round. We still assume at this point that the states sent by $\prov$ are pure. For a discussion about mixed states see Section~\ref{sec:quantumvermixedstates}. First we state a slightly changed theorem.
\begin{theorem}[Quantum Verifier]\label{thm:quantumverifiernoiid}
For every circuit $C$ acting on $n$ qubits, for every $\delta \in (0,\frac13)$ and all $\eta > 0$ sufficiently small there exists an interactive protocol between a quantum verifier $\ver$ and a quantum prover $\prov$ with the following properties. The protocol runs in $N = O(\frac{1}{\eta^2} \log(1/\delta))$ rounds, in each round $\prov$ sends a pure quantum state on $n$ qubits to $\ver$. At the end of the protocol $\ver$ outputs $\bot$ when it rejects the interaction or $x \in \nbits$ when it accepts.

\begin{itemize}
    \item (Completeness) There exists $\prov^{\oracle(*)}$ such that for every $\mathcal{D} \in \mathfrak{D}(n)$ satisfying $\heli(\mathcal{D}, \mathcal{D}_C) \leq \eta$ the following holds.. With probability $1- \delta$ over the randomness in the protocol $\prov^{\oracle(\mathcal{D})}$ succeeds and $x \sim_{\text{i.i.d.}} \mathcal{D}$.
    \item (Soundness) For every $\prov$ that succeeds with probability $\geq 1- \frac\delta2$ we have that with probability $1-\delta$ over the randomness in the protocol $x \sim_{\text{i.i.d.}} \dist^A$ and $\heli(\dist_C, \dist^A) \leq O\left(\eta^{1/4}\right)$. 
\end{itemize}
\end{theorem}

Before going to the proof of the theorem we first state a technical lemma.
\newpage
\begin{lemma}\label{lem:technicaleta1/4}
For every $\eta > 0, k \in \N$, every set of distributions $\dist, \dist_C, \dist_1^A, \dots, \dist_k^A \in \mathfrak{D}(n)$ and every $q_1, \dots, q_k \in [0,1]$ such that $\sum_{i=1}^k q_1 = 1$ the following holds. Let $f(x) = \frac12 (1 + (1-x^2)^2)$. If $\sum_{i =1}^k q_i f(\heli(\dist_C,\dist_{i}^A)) \geq 1 - 50\eta^2$ then 
$$
\heli \left( \sum_{i = 1}^k q_i \dist_i^A, \dist_C  \right) \leq O\left(\eta^{1/4}\right).
$$
\end{lemma}

\begin{proof}
We bound the quantity
\begin{align}
&\sum_{i = 1}^k q_i \heli(\dist_C,\dist_{i}^A ) \nonumber \\
&\leq \sum_{i = 1}^k q_i \left( \heli(\dist_C, \dist_{i}^A) \mathbbm{1}_{\{\heli(\dist_C, \dist_{i}^A) \leq \sqrt{\eta} \}} + \mathbbm{1}_{\{\heli(\dist_C, \dist_{i}^A) > \sqrt{\eta} \}} \right) \nonumber \\
&\leq \sqrt{\eta} + \sum_{i = 1}^k q_i \mathbbm{1}_{\{\heli(\dist_C, \dist_{i}^A) > \sqrt{\eta} \}} \label{eq:funnybound3}  
\end{align}
Let $l = \sum_{i = 1}^k q_i \mathbbm{1}_{\{\heli(\dist_C, \dist_{i}^A) > \sqrt{\eta} \}}$. By definition and the fact that $f(x) \leq 1 - x^2/2$ we have
$$
\left(1-\frac\eta2 \right) l + (1-l) \geq \sum_{i = 1}^k q_i f(\heli(\dist_C,\dist_{i}^A)).
$$
Using the assumption $\sum_{i =1}^k q_i f(\heli(\dist_C,\dist_{i}^A)) \geq 1 - 50\eta^2$ we get $l \leq 100\eta$. Plugging it in \eqref{eq:funnybound3} we get 
\begin{align}
101\sqrt{\eta} &\geq \sqrt{\eta} + 100\eta \nonumber \\
&\geq \sum_{i = 1}^k q_i \heli(\dist_C,\dist_{i}^A ) \nonumber \\
&\geq \sqrt{2} \sum_{i = 1}^k q_i \tv(\dist_C,\dist_{i}^A )  && \text{As } \heli(\mathcal{P}, \mathcal{Q}) \geq \sqrt{2} \tv(\mathcal{P}, \mathcal{Q})\nonumber \\
&\geq \sqrt{2} \tv \left(\dist_C, \sum_{i = 1}^k q_i \dist_i^A \right) && \text{Triangle inequality and identity } \tv(\mathcal{P},\mathcal{Q}) = \frac12 \|\mathcal{P}-\mathcal{Q}\|_1. \label{eq:totalvarbnd3}
 \end{align}
Now we can bound the quantity of interest
\begin{align*}
&\heli \left( \sum_{i = 1}^k q_i \dist_i^A, \dist_C  \right) \\
&\leq
\sqrt{\tv \left(\sum_{i = 1}^k q_i \dist_i^A, \dist_C  \right)} &&\text{By } \heli(\mathcal{P},\mathcal{Q}) \leq \sqrt{\tv(\mathcal{P},\mathcal{Q})} \\
&\leq O\left(\eta^{1/4}\right) && \text{By \eqref{eq:totalvarbnd3}} 
\end{align*}

\end{proof}

\begin{proof}[Proof of Theorem~\ref{thm:quantumverifiernoiid}]
The modified protocol is given in Figure~\ref{fig:quantumverifier_protocol_noiid}. In each run at most one sample is generated. The number of iterations is changed from $O\left(\frac{K}{\eta^2} \log(1/\delta)\right)$ to  $O\left(\frac{1}{\eta^2} \log(1/\delta)\right)$. The biggest change is in the very last step of the protocol, where instead of returning the whole set $S$ we return a random element from $S$. The reason behind this change will hopefully become clear at the end of the proof. 

\begin{figure}
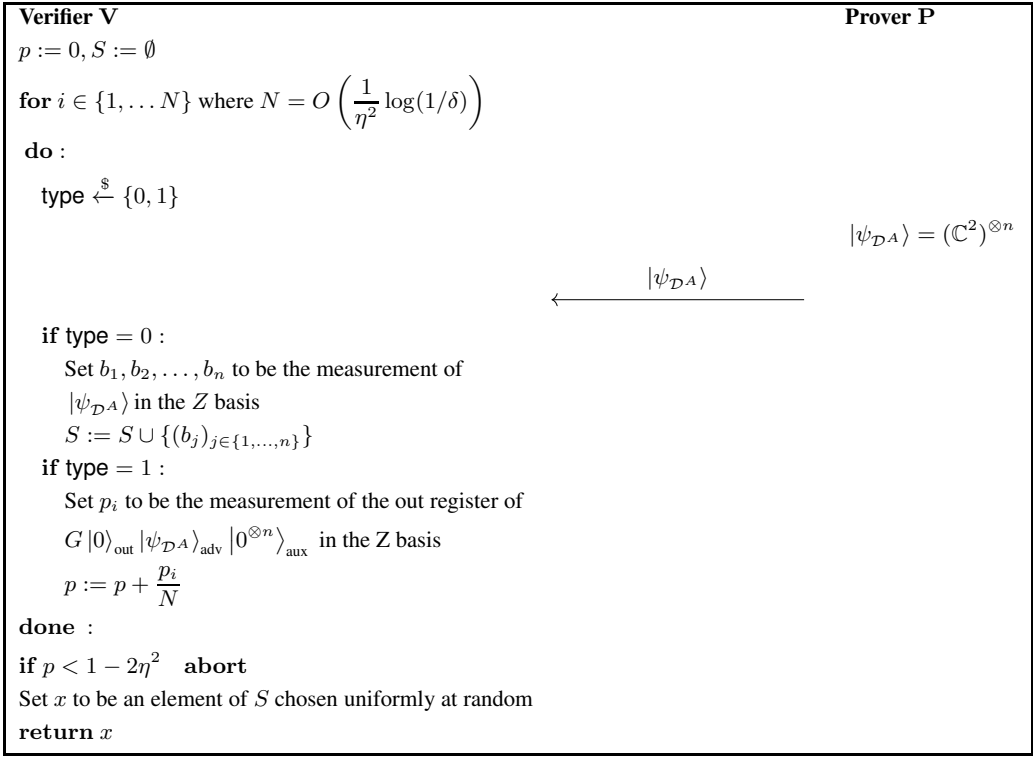

	\scalebox{0.96}{
        \fbox{
            \procedure{}{
                \textbf{Verifier } \ver \>\>  \textbf{Prover } \prov	\\
                p := 0, S := \emptyset \> \> \\
                \pcfor i \in \{1,\dots N\} \text{ where } N = O \left(\frac{1}{\eta^2} \log(1/\delta)\right) \> \>\\ 
                \pcdo:\>\> \\
                \t \textsf{type} \xleftarrow[]{\$} \{0,1\}\> \> \\
                \>\> \ket{\psi_{\dist^A}}  = (\mathbb{C}^2)^{\otimes n} \\
                \> \sendmessageleft*{\ket{\psi_{\dist^A}}} \> \\ 
                \t \pcif \textsf{type} = 0: \\ 
                \t \t \text{Set } b_1, b_2, \dots, b_n \text{ to be the measurement of } \\
                \t \t \ket{\psi_{\dist^A}} \text{in the $Z$ basis}\\
                \t \t S := S \cup \{(b_j)_{j\in\{1 ,\dots , n\}}\}\\
                \t \pcif \textsf{type} = 1: \\ 
                \t \t \text{Set } p_i \text{ to be the measurement of the out register of } \\
                \t \t \comp \ket{0}_\out \ket{\psi_{\dist^A}}_\adver \ket{0^{\otimes n}}_\aux \text{ in the Z basis} \\
                \t \t p := p + \frac{p_i}{N}\\
                \pcdone : \\
                \pcif p < 1 - 2\eta^2 
                \t \pcabort \\
                \text{Set } x \text{ to be an element of } S \text{ chosen uniformly at random} \\
                \pcreturn x
                }
        }
    }
    \caption{The interactive protocol for the model where the verifier has access to a quantum computer and the prover doesn't need to act in an i.i.d. fashion.}
    \label{fig:quantumverifier_protocol_noiid}
\end{figure} 

It suffices to prove the soundness as the completeness proof is analogous to the proof of Theorem~\ref{thm:quantumverifier}.

Assume that $\prov$ sends the states $\ket{\psi_{\dist_1^A}}, \ket{\psi_{\dist_2^A}}, \dots, \ket{\psi_{\dist_N^A}}$ to $\ver$. Let the rounds in which the type is $1$ be $I \subseteq [N]$ and denote $|I|$ by $k$. Then for every $i \in I$ we have that $\ver$ gets a sample according to a Bernoulli variable with parameter $\bra{0^{n}}_\aux \bra{\psi_{\dist_i^A}}_\adver \bra{0}_\out G^\dagger \Pi_{\out}^{(1)} \compcirc \ket{0}_\out \ket{\psi_{\dist_i^A}}_\adver \ket{0^{n}}_\aux$. Thus by Fact~\ref{fact:chernoff} and Corollary~\ref{cor:prelationtohellinger} we have that with probability $1 - \frac\delta2$
\begin{equation}\label{eq:estimatesdf}
\left|p - \frac{1}{k}\sum_{i \in I} f(\heli(\dist_C,\dist_{i}^A)) \right| \leq \eta^2,
\end{equation}
where $f(x) = \frac12(1+ (1-x^2)^2)$. 


$\prov$ succeeds with probability $1 - \frac\delta2$ so by \eqref{eq:estimatesdf} and the union bound we get that with probability $1 - \delta$ 
\begin{equation}\label{eq:fsumlowerbnd}
\frac{1}{k}\sum_{i \in I} f(\heli(\dist_C,\dist_{i}^A)) \geq 1 - 2\eta^2 - \eta^2 \geq  1 - 3\eta^2.
\end{equation}
By Lemma~\ref{lem:technicaleta1/4} we get then
$$
\Prob_I \left[ \heli \left(\frac{1}{k} \sum_{i \in I} \dist_i^A, \dist_C  \right) \geq O(\eta^{1/4}) \right] \leq \delta.
$$
As $I$ and $[N] \setminus I$ have the same distribution we also get
$$
\Prob_I \left[ \heli \left(\frac{1}{N - k} \sum_{i \not\in I} \dist_i^A, \dist_C  \right) \geq O(\eta^{1/4}) \right] \leq \delta.
$$
Finally note that the samples we collected in $S$ came exactly from the distribution $S \sim \Pi_{i \not\in I} \dist_i^A$, so if we choose the sample to return $x$ as a uniformly random element of $S$ then $x \sim \frac{1}{N - k} \sum_{i \not\in I} \dist_i^A$. To conclude note that $\dist^A = \frac{1}{N - k} \sum_{i \not\in I} \dist_i^A$.

\end{proof}

\subsubsection{Prover sending mixed states}\label{sec:quantumvermixedstates}

In this section we explain what happens when instead of sending a pure state $\ket{\psi_{\dist^A}}$, $\prov$ is allowed to send a mixed state $\rho_A$. This means that $\prov$ can prepare a state $\ket{\psi}_{\text{E,F}}$ in a bigger space $(\C^2)^{\otimes n}_{\text{E}} \otimes H_\text{F}$ and send only the $E$ part of the system to $\ver$. We still assume here that $\prov$ acts in an i.i.d. fashion. In this setting the guarantee for soundness will deteriorate (as in Theorem~\ref{thm:quantumverifiernoiid}) to $\heli(\dist_C, \dist^A) \leq O(\eta^{1/4})$ in comparison to $\heli(\dist_C, \dist^A) \leq O(\eta)$ as in Theorem~\ref{thm:quantumverifier}. The slightly changed theorem becomes

\begin{theorem}[Quantum Verifier with Mixed States]\label{thm:quantumverifiermixedstates}
For every circuit $C$ acting on $n$ qubits, for every $\delta \in (0,\frac13), K \in \N$ and all $\eta > 0$ sufficiently small there exists an interactive protocol between a quantum verifier $\ver$ and a quantum prover $\prov$ with the following properties. The protocol runs in $N = O(\frac{K}{\eta^2} \log(1/\delta))$ rounds and in each round $\prov$ sends a (potentially mixed) quantum state on $n$ qubits to $\ver$. At the end of the protocol $\ver$ outputs $\bot$ when it rejects the interaction or it outputs $ S = \{x_1, \dots, x_{|S|} \}$, $x_i \in \nbits$, when it accepts.

\begin{itemize}
    \item (Completeness) There exists $\prov^{\oracle(*)}$ such that for every $\mathcal{D} \in \mathfrak{D}(n)$ satisfying $\heli(\mathcal{D}, \mathcal{D}_C) \leq \eta$ the following holds. With probability $1- \delta$ over the randomness in the protocol  $\prov^{\oracle(\mathcal{D})}$ succeeds, $S \sim_{\text{i.i.d.}} \mathcal{D}^{|S|}$ and $|S| \geq \Omega(K)$.
    \item (Soundness) For every $\prov$ that succeeds with probability at least $\frac23$  we have $S \sim_{\text{i.i.d.}} (\dist^A)^{|S|}$ and $\heli(\dist_C, \dist^A) \leq O(\eta^{1/4})$. 
\end{itemize}
\end{theorem}

\begin{proof}[Proof of Theorem~\ref{thm:quantumverifiermixedstates}]
Note that we only need to verify the soundness property as for the completeness we know that $\prov$ sends pure states. By the ensemble interpretation of density matrices we can express
\begin{equation}\label{eq:densitymatrixexpression}
\rho_A = \sum_{j=1}^k q_i \ket{\psi_{\dist_i^A}}\bra{\psi_{\dist_i^A}},   
\end{equation}
where $\ket{\psi_{\dist_i^A}} \in (\C^2)^{\otimes n}$. This expression is not unique but it will not play a role for us. We observe that measuring $\rho_A$ in the $Z$ basis and collecting a sample is equivalent to collecting a sample from a distribution $\sum_{j=1}^k q_i \dist_i^A$. By Corollary~\ref{cor:prelationtohellinger} we know that the probability of obtaining outcome $1$ when running $\compcirc$ on $\ket{\psi_{\dist_i^A}}$ and measuring out register is equal to the Bernoulli variable with parameter $f(\heli (\dist_C, \dist_i^A))$, for $f(x) = \frac12 (1 + (1-x^2)^2)$. The distribution of measuring the out register when running $\compcirc$ on $\rho_A$ is thus equal to
$$
\sum_{j=1}^k q_i \cdot f(\heli (\dist_C, \dist_i^A)).
$$
By Fact~\ref{fact:chernoff} and the setting of $N$ we have that with probability $1 - \frac\delta2$
$$
\left|p - \sum_{j=1}^k q_i \cdot f(\heli (\dist_C, \dist_i^A))\right| \leq \eta^2
$$

$\prov$ succeeds with probability $\frac23$ so by the union bound and the fact that $\frac13 + \frac\delta2 < 1$ we get that $\sum_{j=1}^k q_i \cdot f(\heli (\dist_C, \dist_i^A)) \geq p - \eta^2 \geq 1 - 3\eta^2$. Application of Lemma~\ref{lem:technicaleta1/4} finishes the proof.
\end{proof}

\subsection{Constant Memory Quantum Verifier}\label{sec4.2}
In section~\ref{sec4.1} we described a protocol in which a verifier $\ver$ can certify that
the distribution of the samples they get from the prover $\prov$ is $\eta$-close to the
distribution of the samples given by the nature. However, this protocol required $\ver$ to 
perform computation on $2n+1$-qubit states, whereas here we assume quantum memory of $\prov$ is constant. 

We proceed by describing a protocol, achieving the same goal, in which $\ver$ can perform operations only on constant number of qubits.\footnote{This protocol is based on a circuit-to-Hamiltonian reduction. The size of this constant
depends on which reduction we use} On a high level $\ver$ wants to outsource the execution of the comparison circuit $\compcirc$ to $\prov$. Intuitively we want $\prov$ to send to $\ver$ a state that certifies execution of $\compcirc$. This is possible thanks to a well-known result called circuit-to-Hamiltonian reduction.

\begin{figure}
    \scalebox{0.81}{
    \fbox{
        \procedure{}{
            \textbf{Verifier } \ver \>\> \textbf{Prover } \prov	\\
            \> \sendmessageright*{H_G}\> \\
            \gamma = 0 , p = 0 \> \> \\
            T' = n + T + 3 \\
            n' = 2n+T'+1 \\
            L = \text{ number of terms of } H_\compcirc\\ 
            \pcfor i \in \{1,\dots N\} \text{ where } N = O\left( \frac{K (n^5 + n^2T^3 +T^5)}{\eta^4}\log(1/\delta)\right) \> \>\\ 
            \pcdo : \>\> \\
            \t t \xleftarrow{\$} \textsf{Terms}(H_C)\> \> \\ 
            \t \textsf{type} \xleftarrow{\$} \{1,2,3\}\> \> \\
            \t n_1 , n_2 , n_3 = 0\\
            \>\> \ket{\phi}_{AB} \in (\C^2)^{\otimes n'}_A \otimes \mathcal{H}_B  \\
            \> \sendmessageleft*{\ket{\phi}_A} \> \\ 
            \t \pcif \textsf{type} = 0: \\
            \t \t \textit{\# measure energy of the state with respect to the Hamiltonian}\\
            \t \t b_i = \text{ measurement of }\ket{\phi}_i \text{ in basis } \sigma_i^B, \t \forall{\sigma_i^B} \in t\\
            \t \t \gamma = \gamma - J_{t} (\Pi_{i \in t}{b_i})\\
            \t \t n_1 = n_1 +1 \\
            \t \pcif \textsf{type} = 1: \\ 
            \t \t \textit{\# obtain a sample}\\
            \t \t b = \text{ measurement of the second qubit of clock in the Z basis}\\
            \t \t  a_1,\dots,a_n = \text{ measurement of aux in the Z basis}\\
            \t\t b' = \text{ measurement of out in the Z basis}\\
            \t\t \pcif b = 0, b' = 0, a_1, \dots, a_n = 0 :\\
            \t\t \t  b_{1} ,\dots, b_{n} = \text{ measurement of adv in the Z basis } \\
            \t\t \t S = S \cup \{(b_i)_{i\in [n]}\}\\
            \t\t n_2 = n_2 + 1 \\
            \t \pcif \textsf{type} = 2: \\ 
            \t \t \textit{\# estimate the output probability}\\
            \t\t b_1 = \text{ measurement of the $T'$-th qubit of clock in the Z basis} \\
            \t\t \pcif b_1 = 1: \\
            \t\t \t r = \text{ measurement of out in the Z basis}\\
            \t\t \t p = p + r\\
            \t\t \t n_3 = n_3 + 1 \\
            \textbf{Done}: \\
            \pcif \frac{\gamma \cdot L}{n_1} > \frac{\eta^2}{2T'^2} \lor \frac{p}{n_3} < 1 - 2\eta^2:\\
            \t \textbf{Abort}
            }
    }   
    }
    \caption{The interactive protocol, in which the verifier collects samples from a distribution close to the desired one. The verifier only requires a single qubit, as they measure one qubit at a time. $H_\compcirc$ is the Hamiltonian corresponding to the comparison circuit, described in Figure~\ref{fig:CS1}. We emphasize that we send $\ket{\phi}$ one qubit a time. Note that when we measure the clock register we use the unary representation of the clock. By writing $\ket{\phi}_{AB} \in (\C^2)^{\otimes n'}_A \otimes \mathcal{H}_B$ and then sending $\ket{\phi}_A$ to $\ver$ we mean that $\prov$ might be sending a mixed state.}
    \label{fig:low_mem_protocol}
\end{figure}

\paragraph{Circuit-to-Hamiltonian reduction.} This reduction was introduced by Kitaev in the late 1990's, see \citet{kitaevbook}. This reduction allows one to reduce a computation of a quantum circuit to estimating the ground energy of a local Hamiltonian. With such a tool in hand $\ver$ can first perform the reduction to create $H_\compcirc$, send a classical description of $H_\compcirc$ to $\prov$, then $\prov$ is supposed to send a low energy state $\ket{\psi}$ of $H_\compcirc$ back to $\ver$, and finally $\ver$ estimates the energy of $\ket{\psi}$ with respect to $H_\compcirc$ to verify that it is indeed of low energy.

For our purposes we need a slight modification of the standard reduction. Due to this fact, here we give an overview of this classical result and point to the differences needed for our setup. We follow the approach from \citet{kitaevbook} and we refer the reader to this book for more details.

The starting point of the reduction is the comparison circuit $\compcirc$\footnote{The reduction can be applied to any circuit but we focus only on the comparison circuit for simplicity.}. Recall that $\compcirc$ acts on three registers: out ($1$ qubit), adv ($n$ qubits), aux ($n$ qubits) and the output of the circuit is obtained by measuring the out register in the $Z$ basis. We want to find an object called a local Hamiltonian $H_\compcirc$. 
\begin{definition}
We say that an operator $H : (\C^2)^{\otimes N} \rightarrow (\C^2)^{\otimes N}$ on $N$ qubits is a $k$-local Hamiltonian if $H$ is expressible as $H = \sum_{r=1}^j H_j$, where each $H_j$ is a Hermitian operator acting on $k$ qubits.
\end{definition}
Our goal will be to define a Hamiltonian that is $5$-local. As mentioned in Section~\ref{sec:overview} $H_\compcirc$ acts on a bigger number of qubits than $\compcirc$ does. More precisely it acts on four registers \textit{clock}, comp = (out, adv, aux) - that is there is an additional register called clock in comparison to registers of $\compcirc$.
The standard reduction defines
\begin{align*}
H_\compcirc = H_{\text{in}} + H_{\text{out}} + H_{\text{prop}} + H_{\text{clock}}.
\end{align*}
The high level idea is to define the terms $H_{\text{in}}, H_{\text{out}}, H_{\text{prop}}, H_{\text{clock}}$ such that $\compcirc$ outputs $1$ with high probability if and only if $H_\compcirc$ has a small eigenvalue. In this case the minimizing vector $\ket{\phi}$ is the so called \textit{history state}
\begin{equation}\label{eq:historystate}
\frac{1}{\sqrt{T'+1}}\sum_{j=0}^{T'} \ket{j}_{\text{clock}} \otimes \compcirc_j \dots \compcirc_1 \ket{0}_{\text{out}}\ket{\psi}_{\text{adv}}\ket{0^{n}}_{\text{aux}},
\end{equation}
where, for every $j$, $\compcirc_j$ is the unitary transformation corresponding to the $j$-th gate in $\compcirc$ and $\ket{j}_{\text{clock}}$ is a state in the clock state space that we will define in detail later. The terms are defined so that they impose penalties to $\bra{\phi} H_\compcirc \ket{\phi}$ whenever $\ket{\phi}$ is far from the history state.

For our purposes we change the reduction by removing the $H_{\text{out}}$ term. By doing that we will be able to say that for every $\ket{\phi}$ such that $\bra{\phi} H_\compcirc \ket{\phi}$ is small there exists $\ket{\psi_{\dist^A}}_{\text{adv}}$ such that $\ket{\phi}$ is close to the history state for $\ket{\psi_{\dist^A}}_{\text{adv}}$. With that property in hand we can then say that if we measure $\ket{\phi}$ in the $Z$ basis then (i) with probability $\Omega(1/T')$ the clock register is equal to $\ket{0}_{\text{clock}}$, the out is equal to $\ket{0}_\out$, the aux register is equal to $\ket{0^n}_\aux$ and the adv register contains a sample from ${\dist^A}$ (ii) with probability $\Omega(1/T')$ the clock register is equal to $\ket{T'}_{\text{clock}}$ and the out register contains a sample from a Bernoulli variable with parameter $p$ such that $p$ is close to the probability of $\compcirc$ outputting $1$ on $\ket{0}_{\text{out}}\ket{\psi_{\dist^A}}_{\text{adv}}\ket{0^{n}}_{\text{aux}}$. Note that we can also write this probability as $ \bra{0^{n}} \bra{\psi_{\dist^A}}_{\text{adv}} \bra{0}_\text{out} G^\dagger \Pi_1^{(1)} \compcirc \ket{0}_{\text{out}}\ket{\psi_{\dist^A}}_{\text{adv}}\ket{0^{n}}_{\text{aux}}$, where $\Pi_s^{(\alpha)}$ is the projection onto the subspace of vectors for which the $s$-th qubit equals $\alpha$. This notation will be useful later.

\paragraph{Overview of the Protocol.} Assuming that the above properties hold we give a high level idea of the protocol defined in Figure~\ref{fig:low_mem_protocol}. In each round of the protocol we perform one of the three types of operations, where the type is chosen uniformly at random (i) we estimate the energy $\bra{\phi} H_\compcirc \ket{\phi}$ (ii) we measure $\ket{\phi}$ in the $Z$ basis and if the clock register is equal to $\ket{0}_\text{clock}$, the out register is equal to $\ket{0}_\out$ and the aux register is equal to $\ket{0^n}_\aux$  then we collect a sample (iii) we measure $\ket{\phi}$ in the $Z$ basis and if the clock register is equal to $\ket{T'}_{\text{clock}}$ we update the estimate for $p$. We run the protocol $\Theta(T)$ rounds thus each of the types will occur $\Omega(T)$ times with high probability and our reduction guarantees that for (ii) we successfully $\Omega(1)$ samples and for (iii) we update the estimate $\Omega(1)$ times. Overall this guarantees that the estimate for $\bra{\phi} H_\compcirc \ket{\phi}$ and $p$ will be accurate and the number of samples collected will be in $\Omega(1)$. Our reduction guarantees moreover that if $\ket{\phi}$ is in fact a low energy state of $H_\compcirc$ then $p$ is close to the probability of $\compcirc$ outputting $1$ on $\ket{0}_{\text{out}}\ket{\psi_{\dist^A}}_{\text{adv}}\ket{0^{n}}_{\text{aux}}$ and the samples we collect come i.i.d. from distribution $\mathcal{D}^A$ that corresponds to $\ket{\psi_{\dist^A}}$. Moreover using Lemma~\ref{lem:similaritymeasure} from $p$ we can estimate $|\langle \psi_{\dist^A} | \psi_{ \dist_C} \rangle |$, recall that $\dist_C$ is the distribution generated by $C$ on $\ket{0^{n}}$ of which we think as being close to $\dist$. As explained in Section~\ref{sec4.1} estimating $|\langle \psi_{\dist^A} | \psi_{ \dist_C} \rangle |$ is enough to guarantee that the distribution from which we collected the samples is close to $\dist$.


For the remainder of this section we first explain the details of the circuit-to-Hamiltonian reduction and then formalize the correctness and soundness requirements and prove the desired properties. 

\subsubsection{Circuit-to-Hamiltonian Reduction}

We start with a quantum circuit $\compcirc$ and want to create a Hamiltonian $H_\compcirc$ with the properties mentioned in Section~\ref{sec4.2}. First we make our goal formal.

\begin{lemma}[Circuit-to-Hamiltonian Reduction]\label{lem:circuittoham}
For every comparison circuit $\compcirc$, for all, sufficiently small, $\e > 0$ there exists an efficiently computable description of a $5$-local Hamiltonian $H_\compcirc$ with $L = O(n+T')$ many terms such that the following conditions hold. Let $\dist^A$ be the distribution of the content of the adv register when measuring $\ket{\phi}$ in the $Z$ basis conditioned on the clock, out and aux registers being all $0$ after measurement. For every $\ket{\phi}$ such that $\bra{\phi} H_\compcirc \ket{\phi} \leq \frac{\e}{T'}$ if we measure $\ket{\phi}$ in the $Z$ basis then
\begin{itemize}
    \item with probability $\in \left[\frac{1-5\e}{T'+1}, \frac{1+5\e}{T'+1}\right]$ the clock register is equal to $\ket{0}_\clock$, the out register is equal to $\ket{0}_\out$, the aux register is equal to $\ket{0^n}_\aux$,
    \item with probability $\in \left[\frac{1-5\e}{T'+1}, \frac{1+5\e}{T'+1}\right]$ the clock register is equal to $\ket{T'}_\clock$ 
    and conditioned on this event the distribution of the out register is a Bernoulli variable with parameter $p$ such that $|p - \bra{0^{n}}_\aux \bra{\psi_{\dist^A}}_\adver \bra{0}_\out G^\dagger \Pi_{\out}^{(1)} \compcirc \ket{0}_\out \ket{\psi_{\dist^A}}_\adver \ket{0^{n}}_\aux| \leq 5 \e T'$.
\end{itemize}
\end{lemma}

\begin{proof}
As we discussed we want to base our reduction on the standard circuit-to-Hamiltonian reduction but drop the $H_\text{out}$ term. We define
\begin{equation}\label{eq:hamiltonian}
H_\compcirc = H_\inp + H_\prop + H_\clock.    
\end{equation}
The term $H_\inp$ corresponds to the condition that, at step $0$, the qubits are in the right state. Formally
\begin{equation}\label{eq:hin}
H_\inp = \ket{0}\bra{0}_\clock \otimes \left( \sum_{j \in \out, \aux} \Pi^{(1)}_j \right) ,
\end{equation}
where by $j \in \out, \aux$ we mean iterating over all the qubits in these registers. Informally speaking, we add a penalty whenever a qubit in registers out or aux is in state $\ket{1}$ while the clock is in state $\ket{0}_\clock$.

The term $H_\prop$ guarantees the propagation of quantum states through the circuit. Formally
\begin{equation}\label{eq:hprop}
H_\prop = \sum_{j=1}^{T'} H_j, 
\end{equation}
$$
H_j = -\frac12 \ket{j}\bra{j-1}_\clock \otimes G_j - \frac12 \ket{j-1}\bra{j}_\clock \otimes G_j^\dagger + \frac12 (\ket{j}\bra{j}_\clock + \ket{j-1}\bra{j-1}_\clock) \otimes I.
$$
We will define $H_\clock$ later. We could realize it with $O(\log(T')$ qubits but then our Hamiltonian would be $O(\log(T'))$-local. But we aim for a $5$-local Hamiltonian. We explain how to address this issue towards the end of this section. Because of this we will assume for now that $H_\clock$ does not appear in \eqref{eq:hamiltonian}.

For the analysis we follow \citet{kitaevbook}. It will be useful to consider a change of basis given by
$$
W = \sum_{j=0}^{T'} \ket{j}\bra{j}_\clock \otimes G_j \dots G_1.
$$
What we mean is that we represent the vector $\ket{\phi}$ in the form $\ket{\phi} = W \ket{\tilde{\phi}}$. Under this change the Hamiltonian is transformed into its conjugate $\widetilde{H}_\compcirc = W^\dagger H_\compcirc W$. Simple calculation verifies that $\widetilde{H}_\inp = H_\inp$ and $\widetilde{H}_\prop = E \otimes I$, where 
$$
E = \begin{pmatrix}
\frac12 & -\frac12 & 0 & 0 & 0 & \\
-\frac12 & 1 & -\frac12 & 0 & 0 &  \\
0 & -\frac12 & 1 & -\frac12 & 0 &  \\
0 & 0 & -\frac12 & 1 & -\frac12 & \\
0 & 0 & 0 & -\frac12 & 1 \\
 & & & & & \ddots
\end{pmatrix}
$$

Let $\ket{\tilde{\phi}}$ be such that $\bra{\tilde{\phi}} \widetilde{H}_\compcirc \ket{\tilde{\phi}} \leq \frac{\e}{T'}$. We will show that it is close to a history state of $\ket{\psi_{\dist^A}}_\adver$. 
Let's write $\tl{\phi} = \sum_{j=0}^{T'} \alpha_j \ket{j}_\clock \ket{\xi_j}_
\comput$, for $\alpha_j \in \R_{\geq 0}$ and $\ket{\xi_0}_\comput = \sum_{s \in \{0,1\}^{n+1}} \beta_s \ket{s[1]}_\out \ket{\psi_s}_\adver \ket{s[2,n+1]}_\aux$, where $s[i,j]$ denotes the substring of $s$ from $i$ to $j$. Then by the fact that $\wt{H}_\prop = E \otimes I$ we have
\begin{align}
\bra{\tl{\phi}} \wt{H}_\prop \ket{\tl{\phi}} 
&=\frac12 \sum_{j=1}^{T'} \|\alpha_{j-1}\ket{\xi_{j-1}}_\comput - \alpha_{j}\ket{\xi_{j}}_\comput\|^2 \nonumber \\
&\geq \frac12 \sum_{j=1}^{T'} |\alpha_{j-1} - \alpha_j|^2, \frac12 \sum_{j=1}^{T'} \min(|\alpha_{j-1})|^2, |\alpha_j|^2) \cdot  \|\ket{\xi_{j-1}}_\comput - \ket{\xi_j}_\comput\|^2  \label{eq:hproppenalty}.  
\end{align}
Note that the bound above gives two inequalities. Thus we get that $\max_{j \in [T']} |\alpha_{j-1} - \alpha_j|^2 \leq \frac{2\e}{T'}$, which combined with the fact that $\sum_{j = 0}^{T'} |\alpha_j|^2 = 1$ gives us that 
\begin{equation}\label{eq:alphascloseto1overt}
\max_{j \in \{0,\dots,T'\}} \left||\alpha_j|^2 - \frac{1}{T'+1} \right| \leq \frac{2\e}{T'}. 
\end{equation}
Using \eqref{eq:alphascloseto1overt} and the bound for $\frac12 \sum_{j=1}^{T'} \min(|\alpha_{j-1})|^2, |\alpha_j|^2) \cdot  \|\ket{\xi_{j-1}}_\comput - \ket{\xi_j}_\comput\|^2$ from \eqref{eq:hproppenalty} we get that for $\e \leq 1$
$$
\sum_{j=1}^{T'} \|\ket{\xi_{j-1}}_\comput - \ket{\xi_j}_\comput\|^2 \leq \frac{\frac{2\e}{T'}}{\frac{1}{T'+1} - \frac{2\e}{T'}} \leq 4\e,
$$
which implies that
\begin{equation}\label{eq:distancebetweenxis}
\|\ket{\xi_{0}}_\comput - \ket{\xi_{T'}}_\comput\|^2 \leq 4\e T'.
\end{equation}

Using the second term from $H_\compcirc$ we also have 
\begin{align}
\bra{\tl{\phi}} \wt{H}_\inp \ket{\tl{\phi}} = \sum_{j=1}^n \sum_{s \in \{0,1\}^n : s[j] = 1} \beta_s^2 \leq \frac{\e}{T'}. \label{eq:hinpenalty}
\end{align}
Note that the distribution corresponding to $\ket{\psi_{0^n}}$ is $\dist^A$. Observe moreover that \eqref{eq:alphascloseto1overt} guarantees that for small enough $\e$ if we measure $\ket{\tl{\phi}}$ in the $Z$ basis then with probability $\in \left[\frac{1-3\e}{T'+1}, \frac{1+3\e}{T'+1}\right]$ the clock register is equal to $\ket{0}_\clock$ and with probability $\in \left[\frac{1-3\e}{T'+1}, \frac{1+3\e}{T'+1} \right]$ the clock register is equal to $\ket{T'}_\clock$. Moreover conditioned on the clock register being $\ket{0}_\clock$ probability of out and aux register being $\ket{0}_\out, \ket{0^n}_\aux$ respectively is, by \eqref{eq:hinpenalty}, lower bounded by $1 - \frac{\e}{T'}$. Thus 
we collect a sample from $\dist^A$ with probability $\in \left[\frac{1 - 3\e}{T'+1} (1 - \frac{\e}{T'}),\frac{1 - 3\e}{T'+1}\right] \subseteq  \left[\frac{1-5\e}{T'+1}, \frac{1+5\e}{T'+1}\right]$.

For the second condition observe that
\begin{align*}
&|p - \bra{0^{n}}_\aux \bra{\psi_{\dist^A}}_\adver \bra{0}_\out G^\dagger \Pi_{\out}^{(1)} \compcirc \ket{0}_\out \ket{\psi_{\dist^A}}_\adver \ket{0^{n}}_\aux| \\
&= |\bra{\xi_{T'}}_\comput W^\dagger \Pi^{(1)}_\out W \ket{\xi_{T'}}_\comput - \bra{0^{n}}_\aux \bra{\psi_{\dist^A}}_\adver \bra{0}_\out G^\dagger \Pi_{\out}^{(1)} \compcirc \ket{0}_\out \ket{\psi_{\dist^A}}_\adver \ket{0^{n}}_\aux| \\
&\leq |\bra{\xi_{0}}_\comput W^\dagger \Pi^{(1)}_\out W \ket{\xi_{0}}_\comput - \bra{0^{n}}_\aux \bra{\psi_{\dist^A}}_\adver \bra{0}_\out G^\dagger \Pi_{\out}^{(1)} \compcirc \ket{0}_\out \ket{\psi_{\dist^A}}_\adver \ket{0^{n}}_\aux| + 4\e T' \\
&\leq \frac{\e}{T'} + 4\e T' \leq 5\e T',
\end{align*}
where in the first inequality we used \eqref{eq:distancebetweenxis} and the fact that the largest eigenvalue of $W^\dagger \Pi^{(1)}_\out W$ is at most of norm $1$ and in the second inequality we used \eqref{eq:hinpenalty} and again the fact that the largest eigenvalue of $W^\dagger \Pi^{(1)}_\out W$ is at most of norm $1$.

\paragraph{Realizing the clock.} As we mentioned we also need to specify how to realize the clock register. The naive implementation would result in a $O(\log(T'))$-local Hamiltonian. To obtain a $5$-local Hamiltonian we use a unary representation. That is we embed the counter space in a larger space in the following way
$$
\ket{j}_\clock \mapsto |\underbrace{1,\dots,1}_j, \underbrace{0,\dots,0}_{T' -j}\rangle.
$$
We need to now change $H_\inp$ and $H_\prop$ to be consistent with this change. But more importantly we need to also penalize incorrect configurations in the clock register. This is what the $H_\clock$ term is responsible for. We refer the readeer to \citet{kitaevbook} for details. The proof of Lemma~\ref{lem:circuittoham} extends naturally to this case. 
\end{proof}

We will need a slight extension of Lemma~\ref{lem:circuittoham} to the case where $\prov$ sends mixed states. For the standard use cases of the reduction this extension is trivial but our purposes require more careful treatment. The difference of our setup in comparison to the standard reduction is that we also collect samples that need to satisfy a specific requirement and this is the reason why the analysis is more involved.

\begin{corollary}[Circuit-to-Hamiltonian Reduction for Mixed States]\label{cor:circuittoham}
For every comparison circuit $\compcirc$, if $\heli(\dist, \dist_C) = \eta$ is sufficiently small then there exists an efficiently computable description of a $5$-local Hamiltonian $H_\compcirc$ with $L = O(n + T')$ many terms such that the following conditions hold. Let $\dist^A$ be the distribution of the content of the adv register when measuring $\rho_A$ in the $Z$ basis conditioned on the clock, out and aux registers being all $0$ after measurement. For every density matrix $\rho_A$ such that $\trace(H_\compcirc \rho_A) \leq \frac{\eta^2}{T'^3}$ if we measure $\rho_A$ in the $Z$ basis then
\begin{itemize}
    \item with probability $\in \left[\frac{1-7\eta}{T'+1}, \frac{1+7\eta}{T'+1}\right]$ the clock register is equal to $\ket{0}_\clock$, the out register is equal to $\ket{0}_\out$, the aux register is equal to $\ket{0^n}_\aux$,
    \item with probability $\in \left[\frac{1-7\eta}{T'+1}, \frac{1+7\eta}{T'+1}\right]$ the clock register is equal to $\ket{T'}_\clock$ and if conditioned on this event the distribution of the out register is a Bernoulli variable with parameter $p \geq 1 - 3\eta^2$ then $\heli(\dist_C, \dist^A) \leq O(\eta^{1/4})$.
\end{itemize}
\end{corollary}

\begin{proof}
Let $\e = \frac{\eta^2}{T'^2}$. By the ensemble interpretation of density matrices we can express
$$
\rho_A = \sum_{i=1}^k q_i \ket{\phi_i} \bra{\phi_i}_\comput.
$$
Thus we can write
$$
\sum_{i=1}^k q_i \bra{\phi_i} H_\compcirc \ket{\phi_i} \leq \frac{\e}{T'}.
$$
By Markov inequality we have 
\begin{equation}\label{eq:markov1}
\sum_{i=1}^k q_i \mathbbm{1}_{\left\{\bra{\phi_i} H_\compcirc \ket{\phi_i} > \frac{\sqrt{\e}}{T'}\right\}} \leq \frac{\sqrt{\e}}{T'}.
\end{equation}
For $i \in [k]$ let $\dist_i^A$ be the distribution of contents of adv conditioned on clock, out, and aux registers being all $0$ when measuring $\ket{\phi_i}$ in the Z basis. Note that for all $i$ such that $\bra{\phi_i} H_\compcirc \ket{\phi_i} \leq \frac{\sqrt{\e}}{T'}$ Lemma~\ref{lem:circuittoham} guarantees that $\dist_i^A$ satisfies the conditions of the reduction. 

To see the the first condition note that by \eqref{eq:markov1} we get that the probability that the clock register is $\ket{0}_\clock$ is $\in \left[\frac{1-5\e - 2\sqrt{\e}}{T'+1}, \frac{1+5\e + 2\sqrt{\e}}{T'+1}\right] \subseteq \left[\frac{1-7\sqrt{\e}}{T'+1}, \frac{1+7\sqrt{\e}}{T'+1}\right] \subseteq \left[\frac{1-7\eta}{T'+1}, \frac{1+7\eta}{T'+1}\right]$. Same bound on probability holds also for the clock register being equal to $\ket{T'}_\clock$.

For $i \in [k]$ let $p_i$ be the probability of obtaining outcome $1$ in the out register when measuring $\ket{\phi_i}$ in the $Z$ basis conditioned on clock register being in state $\ket{T'}_\clock$. Then for the second condition observe that

\begin{align}
p &= \sum_{i=1}^k q_i p_i \nonumber \\
&\leq \sum_{i=1}^k q_i p_i \mathbbm{1}_{\left\{\bra{\phi_i} H_\compcirc \ket{\phi_i} \leq \frac{\sqrt{\e}}{T'}\right\}} + \frac{\sqrt{\e}}{T'} \nonumber \\
&\leq \sum_{i=1}^k q_i \bra{0^{n}}_\aux \bra{\psi_{\dist_i^A}}_\adver \bra{0}_\out G^\dagger \Pi_{\out}^{(1)} \compcirc \ket{0}_\out \ket{\psi_{\dist_i^A}}_\adver \ket{0^{n}}_\aux \mathbbm{1}_{\left\{\bra{\phi_i} H_\compcirc \ket{\phi_i} \leq \frac{\sqrt{\e}}{T'}\right\}} + 6\sqrt{\e}T' \nonumber \\
&\leq \sum_{i=1}^k q_i \bra{0^{n}}_\aux \bra{\psi_{\dist_i^A}}_\adver \bra{0}_\out G^\dagger \Pi_{\out}^{(1)} \compcirc \ket{0}_\out \ket{\psi_{\dist_i^A}}_\adver \ket{0^{n}}_\aux  + 6\sqrt{\e}T' + \frac{\sqrt{\e}}{T'} \nonumber  \\
&\leq \sum_{i = 1}^k q_i f(\heli(\dist_C, \dist_i^A)) + 7\sqrt{\e}T' \label{eq:boundonp}
\end{align}
where in the first inequality we used \eqref{eq:markov1}, in the second inequality we used properties of $\dist_i^A$ guaranteed by Lemma~\ref{lem:circuittoham}, 
in the fourth we used Corollary~\ref{cor:prelationtohellinger}.

By \eqref{eq:boundonp} and the assumption $p \geq 1 - 3\eta^2$ we get that
$$
\sum_{i = 1}^k q_i f(\heli(\dist_C, \dist_i^A))  \geq 1 - 3\eta^2 -7\sqrt{\e}T' \geq 1 - 10\eta^2,
$$
where in the last inequality we used that $\e= \frac{\eta^2}{T'^2}$. We conclude by applying Lemma~\ref{lem:technicaleta1/4}.

\end{proof}

\subsubsection{Correctness of the Protocol}

Recall that protocol from Figure~\ref{fig:low_mem_protocol} builds upon the protocol from Figure~\ref{fig:quantumverifier_protocol}. Now $\ver$, instead of running $\compcirc$ itself, outsources its execution to $\prov$. On a high level correctness of this new protocol is a consequence of correctness of the quantum verifier protocol (Theorem~\ref{thm:quantumverifier}) and circuit-to-Hamiltonian reduction (Lemma~\ref{lem:circuittoham}). One, however, needs to be careful as the guarantees about the protocol will change slightly and some details in the proof need to be verified.

\begin{lemma}\label{lem:threets}
    Let $n_1,n_2,n_3$ be the number of times each type occurs in protocol from Figure~\ref{fig:low_mem_protocol}. If $N = \Omega(\log(1/\delta))$ then $\Prob[n_1,n_2,n_3 > \frac{N}{6}] \geq 1 - \delta$. 
\end{lemma}
\begin{proof}
For $b \in \{1,2,3\}$, $n_b$ can be seen as sum of random Bernoulli variables $\{x_i\}_{i\in [N]}$ with parameter $1/3$. Then by Fact~\ref{fact:chernoff} we get that $\Prob[|\frac{n_b}{N} - \frac13| > \frac16] \leq 2e^{-\frac{N}{72}} \leq \frac\delta3$. We finish by applying the union bound to the error events.

\end{proof}

\begin{lemma}\label{constant-mem-estimate}
Let $\rho_A$ be the reduced density of the first $n'$ qubits of $\ket{\phi}_{AB}$, $\gamma$, $p$, $n_1,n_2,n_3$, $S$ be as in the protocol defined in Figure~\ref{fig:low_mem_protocol}. Let $p^*, q^*$ and $\lambda$ be defined as, 
\begin{align*}
        \lambda &= \trace(H_\compcirc \rho_A),  \\
        q^* &= \trace(\ket{0} \bra{0}_{\clock} \otimes \ket{0} \bra{0}_{\out} \otimes \ket{0^n}\bra{0^n}_\aux \rho_A), \\
        p^* &= \frac{\trace(\ket{T'}\bra{T'}_{\clock} \otimes \ket{1}\bra{1}_{\out} \rho_A )}{\trace( \ket{1}\bra{1}_{\out} \rho_A)}.
\end{align*}
We define the event $\mathcal{F}$ to be $\left| \frac{\gamma \cdot L}{n_1} - \lambda \right| \leq \e, \left| \frac{|S|}{n_2} - q^* \right| \leq \e \left| \frac{p}{n_3} - p^* \right| \leq \e  $. If $N = \Omega( \frac{n^2 +T'^2}{\e^2}\log(1/\delta))$ then $\Prob[\mathcal{F}] \geq 1 - \delta$. 
\end{lemma}

\begin{note}
We note that as explained in Section~\ref{sec:prelim} when we write $\ket{0} \bra{0}_{\clock} \otimes \ket{0} \bra{0}_{\out} \otimes \ket{0^n}\bra{0^n}_\aux$ we really mean $\ket{0} \bra{0}_{\clock} \otimes \ket{0} \bra{0}_{\out} \otimes I_\adver \otimes \ket{0^n}\bra{0^n}_\aux$. We omit the $I$'s for simplicity of notation.
\end{note}

\begin{proof}
Note that for every term $t \in H_\compcirc$ we have $|J_t| \leq 1$. Then if $n_1 = \Omega(\frac{L^2}{\e^2} \log{\frac{1}{\delta}})$ then Fact~\ref{fact:chernoff} guarantees that $\Prob[|\frac{\gamma \cdot L}{n_1}- \lambda| >\e] \leq \delta$. 

Next we define Bernoulli variables $\{s_i\}_{i \in [n_2]}$ to indicate whether $|S|$ increases in a given round, i.e. $|S| = \sum_{i = 1}^{n_2}s_i$. By definition $\mu = E[s_i] = \trace(\ket{0} \bra{0}_{\clock} \otimes \ket{0} \bra{0}_{\out} \otimes \ket{0^n}\bra{0^n}_\aux \rho_A)$. Using Fact~\ref{fact:chernoff} we get that if $n_2 = \Omega(\frac{1}{\e^2}\log(1/\delta))$ then $\Prob \left[\left|\frac{|S|}{n_2}- q^* \right| >\e \right] \leq \delta$. The exact same argument can be used for $\frac{p}{n_3}$.


To conclude we note that, by the union bound, if $n_1 = \Omega( \frac{L^2}{\e^2}\log(1/\delta))$ and $n_2,n_3 = \Omega(\frac{1}{\e^2} \log(1/\delta))$ then $\Prob[\mathcal{F}] \geq 1-\delta$. By Lemma~\ref{lem:threets} and the union bound we get that if $N = \Omega( \frac{L^2}{\e^2}\log(1/\delta))$ then $\Prob[\mathcal{F}] \geq 1-\delta$. As Lemma~\ref{lem:circuittoham} guarantees that $L = O(n + T')$ we can also set $N = \Omega( \frac{n^2 +T'^2}{\e^2}\log(1/\delta))$.

\end{proof}

Intuitively Lemma \ref{constant-mem-estimate} guaranties that with a high probability, the estimates $\frac{\gamma \cdot L}{n_1} , \frac{|S|}{n_2}, \frac{p}{n_3}$ are accurate enough. With that fact in hand we proceed by stating the main theorem of this section. 

\begin{theorem}[Constant Memory Quantum Verifier]\label{thm:constantmemoryverifier}
For every circuit $C$ acting on $n$ qubits, with $T$ gates, for every $\delta \in (0,\frac13), K \in \N$ and all $\eta > 0$ small enough there exists an interactive protocol between a verifier with constant quantum memory $\ver$ and a quantum prover $\prov$ with the following properties. The protocol runs in $N = O\left( \frac{K \cdot (n^5 + n^2T^3 +T^5)}{\eta^4}\log(1/\delta)\right)$ rounds, in each round $\prov$ sends a (potentially mixed) quantum state on $O(n+T)$ 
qubits to $\ver$. At the end of the protocol $\ver$ outputs $\bot$ when it rejects the interaction or $S = \{x_1, \dots, x_{|S|} \}$, where $x_i \in \nbits$, when it accepts.

\begin{itemize}
    \item (Completeness) There exists $\prov^{\oracle(*)}$ such that for every $\mathcal{D} \in \mathfrak{D}(n)$ satisfying $\heli(\mathcal{D}, \mathcal{D}_C) \leq \eta$ the following holds. With probability $1- \delta$ over the randomness in the protocol $\prov^{\oracle(\mathcal{D})}$ succeeds, $S \sim_{\text{i.i.d.}} \mathcal{D}^{|S|}$ and $|S| \geq \Omega(K)$.
    \item (Soundness) For every $\prov$ that succeeds with probability at least $\frac23$ we have $S \sim_{\text{i.i.d.}} (\dist^A)^{|S|}$ and $\heli(\dist_C, \dist^A) \leq O(\eta^{1/4})$. 
\end{itemize}
\end{theorem}

\begin{proof}
We first address completeness of the protocol and then move to soundness.
\paragraph{Completeness.} Recall that the $\prov$ that was guaranteed to exist in Theorem~\ref{thm:quantumverifier} was just sending state $\ket{\psi_\dist}_\adver$ to $\ver$. Recall that we denote by $T' = n + T + 3$ the number of gates in $\compcirc$ and by $n'$ the number of qubits that are sent by $\prov$ in each round. As we discussed the natural extension of this strategy to the constant memory model is for $\prov$ to prepare the history state $\ket{\phi_\dist}_\comput$ of $\ket{\psi_\dist}_\adver$ and send it to $\ver$. As $N = O\left( \frac{K \cdot (n^2T'^3 +T'^5)}{\eta^4}\log(1/\delta)\right) = O\left( \frac{K \cdot (n^5+ n^2T^3 +T^5)}{\eta^4}\log(1/\delta)\right)$ we get by Lemma~\ref{constant-mem-estimate} that with probability $1 - \delta$ 
\begin{itemize}
\item the estimate of the energy $\frac{\gamma \cdot L}{n_1} \leq \frac{\eta^2}{4T'^3}$  as $\bra{\phi_\dist} H_\compcirc \ket{\phi_\dist} = 0$, 
\item $|S| = \Omega(K)$ as in this case $\trace(\ket{0} \bra{0}_{\clock} \otimes \ket{0} \bra{0}_{\out} \otimes \ket{0^n}\bra{0^n}_\aux \rho_A)$, which is the probability of getting a sample if the type is $1$ is equal to $\bra{\phi_\dist} \Pi_\clock^{(T')} \ket{\phi_\dist} = \frac{1}{T'+1}$, 
\item $p \geq \frac{\bra{\phi_\dist} \Pi_\clock^{(0)} \Pi_\out^{(1)} \ket{\phi_\dist}}{\bra{\phi_\dist} \Pi_\clock^{(0)} \ket{\phi_\dist}} - \frac{\eta^2}{4} \geq f(\heli(\dist_C, \dist)) - \frac{\eta^2}{4} \geq 1 - 2\eta^2$, thus the two checks are verified and the interaction is accepted. By definition $S \sim_{\text{i.i.d.}} (\dist)^{|S|}$. Thus completeness is verified. 
\end{itemize}
\paragraph{Soundness.} We follow the structure of the proof of Theorem~\ref{thm:quantumverifiermixedstates}, which is the analog of this theorem for a fully quantum verifier. Let $\rho_A$ be the density matrix representing the state sent by $\prov$. By Lemma~\ref{constant-mem-estimate} we know that with probability $1 -\delta/2$ the energy estimate is within an additive error of $\frac{\eta^2}{4T'^3}$ and $p$ is estimated within an additive error of $\frac{\eta^2}{4}$. So as $\prov$ succeeds with probability $\frac23$ then by the union bound and the fact that $\frac13 + \frac\delta2 < 1$ we get that $\trace(H_\compcirc \rho_A) \leq \frac{\eta^2}{2T'^3} + \frac{\eta^2}{4T'^3} = \frac{\eta^2}{T'^3}$ and $p \geq 1 - 2\eta^2 - \frac{\eta^2}{4} \geq 1 - 3\eta^2$. With that we can apply Corollary~\ref{cor:circuittoham} and conclude that $\heli(\dist^A, \dist_C) \leq O(\eta^{1/4})$. 
\end{proof}

\subsection{Classical Verifier}\label{sec4.3}
Now we are ready to move to the last model we consider in this work, namely the one where $\ver$ is fully classical and the communication is also classical. Recall that in Section~\ref{sec4.2} we designed the protocol by forcing $\prov$ to send to $\ver$ a history state $\ket{\phi_{\dist^A}}_\comput$ corresponding to a distribution satisfying $\heli(\dist_C, \dist^A) \leq O(\eta^{1/4})$. To extend this protocol to the classical model we first force $\prov$ to commit to a state $\rho$, a state that will in some sense correspond to $\ket{\phi_{\dist^A}}_\comput$ and then force $\prov$ to measure this state in the basis chosen by $\ver$. By making the prover to measure his qubits honestly we get a version of constant quantum memory Protocol (Figure~\ref{fig:low_mem_protocol}) in which all the quantum computation is done on the prover side and the verifier and the communication is completely classical.


To achieve our goal we will use cryptographic tools. As the protocol will rely on hardness of computational problems, our soundness results will only address provers who are computationally bounded, namely only provers in the QPT class. Recall that the QPT class is defined as follows. There exists a classical algorithm running in $\text{poly}(\lambda)$ time that for an input size $1^{\lambda}$, generates the prover's circuit of size $\text{poly}(\lambda)$. 

Next we give a high level overview of the protocol.
An honest prover $\prov$ is given a local Hamiltonian correspoding to $G$ and computes the ground state of the Hamiltonian, i.e. the history state $\ket{\phi_{\hist}}$. Later the prover is asked to commit to this state before the protocol proceeds with the interactive stage, in which the prover is asked to measure qubits of the state he has committed to either in computational or the Hadamard basis, and send the outcomes to the verifier. At each iteration, the verifier decides to do one of the following 3,
\begin{itemize}
    \item estimate the energy of the state the prover has committed to,
    \item estimate the probability of the output of the circuit being 1,
    \item collect a sample from the distribution corresponding to the prover's state.
\end{itemize}

The description of this protocol is given in Figure \ref{fig:classical_protocol}. We note that the results presented in this section heavily rely on \citet{mahadev}. Some of the technical lemmas are not proven here. We refer the reader to \citet{mahadev} for said proofs.

\begin{note} We stress that with this protocol one can only retrieve samples from measurements done in the Z basis. The distribution of samples collected in the protocol when $\ver$ asks for the X basis measurements are \textbf{not} in general equal to the distribution of measuring the state of $\prov$ in the X basis. This means that if our protocol required samples from the distribution corresponding to the X measurements it is not clear if it could be realized in the fully classical model.
\end{note}

Similar to \cite{mahadev} we require a more refined version of the circuit-to-Hamiltonian reduction, namely we require our Hamiltonians to be $2$-local and of the form $\sum_{i,j} -\frac{J_{i,j}}{2} (\sigma_{X,i}\sigma_{X,j} + \sigma_{Z,i}\sigma_{Z,j})$. 

\begin{theorem}[\cite{2localhamiltonian}]
For any integer $n \geq 1$ there exists $n' = poly(n)$, $a(n)$ and $\delta \geq 1/poly(n)$ such that given a $T$-gate quantum circuit $\compcirc$, there exists an efficiently computable real-weighted Hamiltonian $H_\compcirc$ in $XX-ZZ$ form, such that, 
\begin{itemize}
    \item (\textit{completeness}) If $\compcirc$ accepts x with probability at least $2/3$, then $\lambda_0(H_\compcirc) \leq a$.
    \item (\textit{soundness}) If $\compcirc$ accepts x with probability at most $1/3$, then $\lambda_0(H_\compcirc) > a+ \delta$.
\end{itemize}
\end{theorem}

As proved in Section~\ref{sec4.2}, by modifying the standard circuit-to-Hamiltonian reduction, we can show that "for any $\ket{\phi}$ such that $\bra{\phi} H_\compcirc \ket{\phi} < \e$ the distribution of the measurement outcome of the first qubit of $\ket{\phi}$ (conditioned on the clock register being $T'$) is $\e$ close to the distribution of what $G$ would output". For the sake of simplicity we skip reproving this statement for the 2-local Hamiltonian. 

\begin{figure}
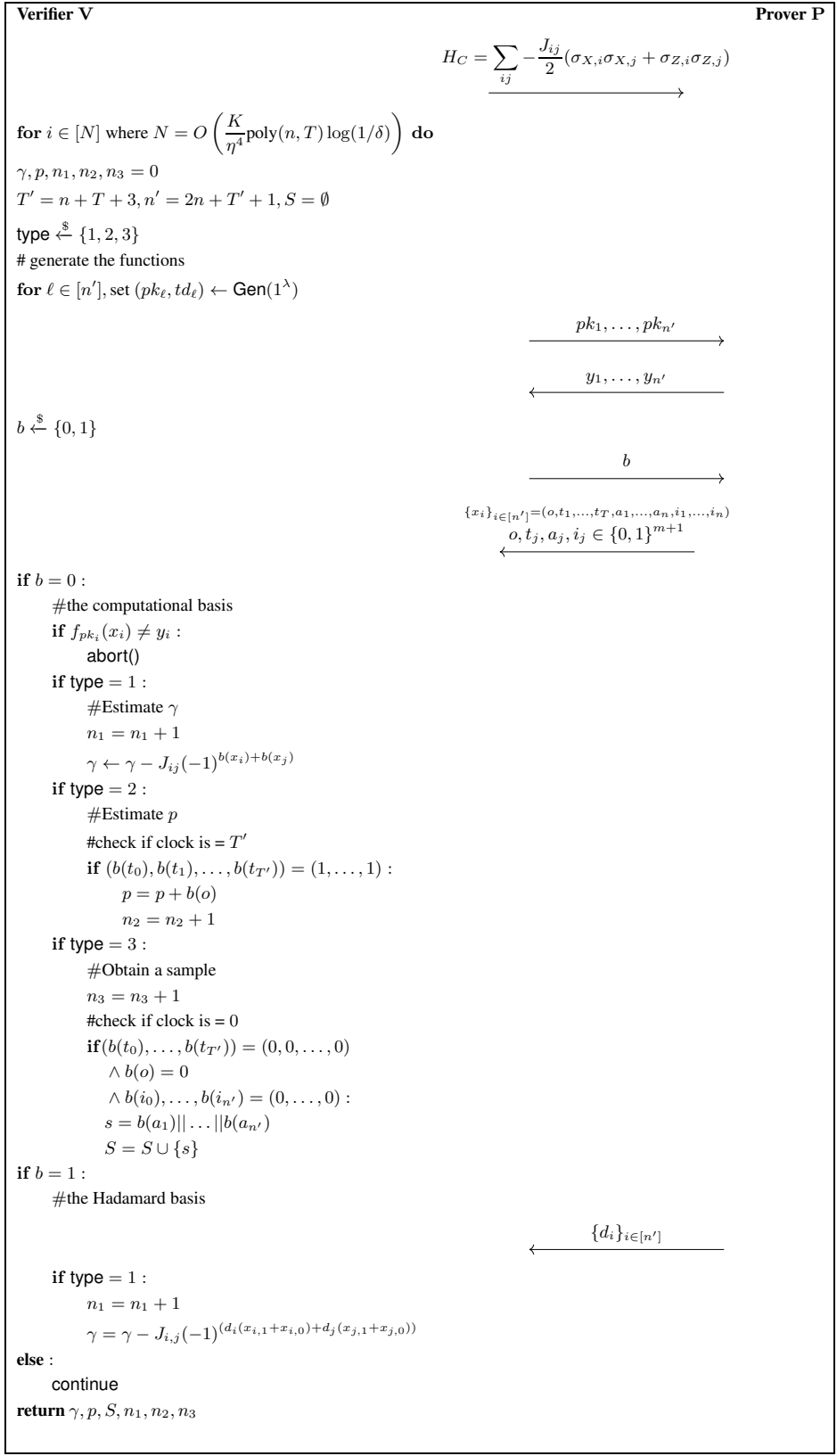

    \centering
        \scalebox{0.83}{
        \fbox{
        \procedure{}{
            \textbf{Verifier } \ver \>\> \textbf{Prover } \prov	\\
            \> \sendmessageright*{H_{C} = \sum_{ij} -\frac{J_{ij}}{2} (\sigma_{X,i} \sigma_{X,j} + \sigma_{Z,i} \sigma_{Z,j})}\> \\
            \pcfor i \in [N] \text{ where } N= O \left(\frac{K}{\eta^4}\poly(n,T)\log(1/\delta)\right) \pcdo \>\> \\
            \gamma , p , n_1,n_2,n_3  = 0 \\
            T' = n+T+3, n' = 2n+T'+1, S = \emptyset\\
            \textsf{type} \xleftarrow{\$} \{1,2,3\}\>\>\\ 
            \text{\# generate the functions}\>\> \\
            \pcfor \ell \in [n'], \text{set } (pk_\ell , td_\ell) \gets \textsf{Gen}(1^\lambda)\\
           \> \sendmessageright*{pk_1,\dots,pk_{n'}}\\
           \> \sendmessageleft*{y_1,\dots,y_{n'}}\> \\
           b \xleftarrow{\$} \{0,1\}\\ 
           \> \sendmessageright*{b} \\
           \> \sendmessageleft*{\stackrel{\{x_i\}_{i \in [n']} = (o , t_1,\dots,t_T, a_1,\dots , a_{n}, i_1 ,\dots, i_n)}{o,t_j,a_j,i_j \in \{0,1\}^{m+1}}}\\
           \pcif b = 0: \\
           \t \t \#\text{the computational basis}\\
           \t \t \pcif f_{pk_i}(x_i) \neq y_i:\\
           \t \t \t \t \textsf{abort()}\\
           \t \t \pcif \textsf{type} = 1: \\
           \t \t \t \t \#\text{Estimate } \gamma\\
           \t \t \t \t n_1 = n_1 + 1 \\
           \t \t \t \t \gamma \gets \gamma - J_{ij}(-1)^{b(x_i) + b(x_j)}\\
           \t \t \pcif \textsf{type} = 2: \\
           \t \t \t \t \#\text{Estimate } p \\
           \t \t \t \t \text{\#check if clock is = }T'\\
           \t \t \t \t \pcif (b(t_0),b(t_1),\dots,b(t_{T'})) = (1,\dots,1):\\
           \t \t \t \t \t \t p = p+ b(o)\\
           \t \t \t \t \t \t  n_2 = n_2 + 1 \\
            \t \t \pcif \textsf{type} = 3:\\
            \t \t \t \t \#\text{Obtain a sample}\\
            \t \t \t \t n_3 = n_3 + 1 \\
            \t \t \t \t \text{\#check if clock is = }0\\
            \t \t \t \t \pcif [(b(t_0),\dots ,b(t_{T'})) = (0,0,\dots,0)\\
            \t \t \t \t \t \land b(o) = 0 \\
            \t \t \t \t \t \land b(i_0),\dots,b(i_{n'}) = (0,\dots,0)]:\\
            \t \t \t \t \t s = b(a_1) ||\dots ||b(a_{n'})\\
            \t \t \t \t \t S = S \cup \{s\}\\
            \pcif b = 1: \\
            \t \t \#\text{the Hadamard basis}\\
           \> \sendmessageleft*{\{d_i\}_{i \in [n']}}\\
           \t \t \pcif \textsf{type} = 1: \\
           \t \t \t \t n_1 = n_1 + 1 \\
           \t \t \t \t \gamma = \gamma - J_{i,j} (-1)^{(d_i(x_{i,1}+x_{i,0}) + d_j(x_{j,1}+x_{j,0}))}\\
           \textbf{else}:\\
           \t \t \textsf{continue}\\
          \textbf{return }\gamma , p , S, n_1 , n_2 , n_3\\
            }}}
             \caption{The description of the classical verifier protocol. Notice that $x_i$ values are $m+1$ bits long each, e.g. $o$ contains the measurement outcome of the output register, plus the remaining $m$ bits of the input to $f_{pk_1}$.} 
    \label{fig:classical_protocol}
\end{figure}

We proceed by stating the completeness and soundness properties of this protocol and providing a proof sketch.

\subsubsection{Cryptographic Assumptions, Claw-Free Functions}
In this section we review the cryptographic assumptions that the soundness of our protocol relies on. 

The protocol starts with $\ver$ sending description of $n'$ functions $f_{pk_1},\dots,f_{pk_{n'}}$ to $\prov$. These functions are a family of 2-to-1 functions called \textit{claw-free}. Intuitively the prover is asked to commit to an image $y_i$ for each $f_{pk_i}$. Let us denote the two preimages of $y_i$ under $f_{pk_i}$ by $x_{i,0}$ and $x_{i,1}$. Based on the challenge bit, the prover is asked to either reveal one of $x_{i,b}$ values or reveal $d_i$ such that $d_i(x_{i,0} + x_{i,1}) = 0$. We also assume that for all $i$ there exists an efficiently computable function $b_i$, which given $x_{i,b}$ outputs the bit $b$. The existence function $b$ is inherent in most claw-free family constructions. 

For an efficiently computable 2-to-1 function $f$, $\prov$ can do the aforementioned task as follows, 
\begin{enumerate}
    \item $\prov$ prepares the state $\frac{1}{\sqrt{2^{m+1}}}\sum_{x\in \{0,1\}^{m+1}} (\ket{x})$
    \item evaluates $f_{pk}$ on the super-position to get the state, $\frac{1}{\sqrt{2^{m+1}}}\sum_{y\in \{0,1\}^{m}} (\ket{x_{0,y}} + \ket{x_{1,y}})\ket{y}$ 
    \item measures the last register to get $y^*$, the state after this step would be $\frac{1}{\sqrt{2}} (\ket{x_{0,y^*}} + \ket{x_{1,y^*}})$ and send $y^*$ to the verifier
    \item \textsf{if} the challenge is $0$, measures the state in the computational basis to obtain one of the preimages $x_{b,y^*}$, and if the challenge is $1$, measures in the Hadamard basis to obtain $d$ s.t. $d(x_{0,y^*} + x_{1,y^*}) = 0$
\end{enumerate}

The cryptographic property that we want to capture with our family of functions is that, although based on the challenge bit, $\prov$ can either return one preimage $x_b$ or $d$ s.t. $d(x_1+x_2) = 0$ for $b\in \{0,1\}$ and $f^{-1}(y) = \{x_0,x_1\}$, but they should not succeed in both tasks simultaneously.
 
\begin{definition}[Adaptive hardcore bit property]\label{def:hdb}
For parameter $\lambda$, a family of functions $\{f_{pk}\}_{pk}:\{0,1\}^{m(\lambda)} \to \{0,1\}^{m(\lambda)-1}$, is called an adaptive hardcore family if, \footnote{$m$ is a polynomial} 
\begin{enumerate}
    \item for all $pk$,  $f_{pk}$ is 2-to-1,
    \item there exists a classically, efficiently computable bijective function, $b_{pk,y}:f^{-1}(y) \to \{0,1\}$
    \item $\forall pk$, $\forall \prov \in \textsf{QPT}$, if $(y,x^*,d) \gets \prov(pk)$, let $\{x_0 , x_1\} = f_{pk}^{-1}(y)$ then $|1/2 - \Prob[f(x^*) = y \land d(x_0 + x_1) = 0]| \leq ngl(\lambda)$.
\end{enumerate}
\end{definition}

For the sake of convenience we only consider a specific construction of claw-free familes. We change the second requirement of definition~\ref{def:hdb} to,

\; $2'.$ For any $pk$ and any $y \in \text{Img}(f_{pk})$, the preimages of $y$ take the form $(b,x_b)$, meaning that $b(x)$ is actually the first bit of $x$. 

\subsubsection{The Canonical Isometry}\label{sec:canonical}
In order to prove $\gamma$ is an accurate estimate of the energy using a similar argument to Lemma~\ref{constant-mem-estimate}, we have to prove the $\mathbb{E}[\gamma] = \trace(H_\compcirc \rho)$, where in a sense $\rho$ is the prover's state. Letting $(X_i,Z_i)$ be the observables of $\prov$ which determin the value of the $i^{th}$ response, we require an isometry which \textit{teleports} these observables to $\sigma_{X,i} , \sigma_{Z,i}$ as the Hamiltonian is penalizing the bad configurations of the state with respect to $\sigma_{X,i} , \sigma_{Z,i}$. Let us assume the prover's state is in a hilbert space $H \otimes H_{\text{env}}$, where he might share some entanglement with the enviroment.

 Based on how the estimates are updated in the protocol the natural way to define the observables that $\prov$ measures, would be,\footnote{here we will use $b$ by absuing the notation instead of $(b_i(x_i))_{i \in [n]}$}
 
 \begin{center}
     $Z(a) = \sum_{x_1, \dots , x_{n'}\in \{0,1\}^{m}} (-1)^{b(x)\cdot a}\ket{x_1}\bra{x_1}\otimes \dots \ket{x_{n'}}\bra{x_{n'}}\otimes I_P$\\
     $X(a) = \sum_{d_1,\dots,d_{n'} } (-1)^{\sum a_i(d_i (x_{i,0}+ x_{i,1}))} U^{\dagger} (\ket{d_1}\bra{d_1} \otimes \dots \ket{d_{n'}}\bra{d_{n'}} \otimes I_P) U$
 \end{center}
Basically in our modelling of actions of $\prov$, if the challenge bit is $b=0$, $\prov$ measures his state in the computational basis in order to get the preimages, and if $b=1$, he applies an arbitrary unitary $U$, followed by a Hadamard measurement to retrieve the $d$ values. 

As mentioned before we would like to see these obserables as $\sigma_X$ and $\sigma_Z$ up to some isometry. The canonical choice of isometry here is $V: H \to (\mathbb{C}^2)^{\otimes n'} \otimes (\mathbb{C}^2)^{\otimes n'} \otimes H$ defined as,
\begin{center}
    $V\ket{\psi} = (\frac{1}{2^{n'}} \sum_{a,b} I\otimes \sigma_X(a)\sigma_Z(b) \otimes X(a)Z(b))\ket{\phi^+}^{\otimes n'}\ket{\psi}$
\end{center}
Where $\ket{\phi^+}$ is an EPR pair. 
\begin{definition}[Extracted Qubits]
Let $\prov$ be a prover playing in Protocol~\ref{fig:classical_protocol}, $X,Z$ defined and let $V$ be the canonical isometry sending $(X,Z)$ to $(\sigma_X , \sigma_Z)$. Let $\ket{\phi} \in H \otimes H_{\text{env}}$ be the state of the prover after sending the $y_i$ values. We call the reduced density of $V\ket{\phi}$ on $H$ the \textit{extracted qubits} of the prover, and we denote it by $\rho$. 
\end{definition}

\begin{fact}[\cite{mahadev}]
Let $\prov$ be any $\textsf{QPT}$ prover, $\rho$ be their extracted qubit. We have, \footnote{$\sigma_W(a) = \Pi_{i \text{ s.t. }a_i = 1} \sigma_{W,i}$ the $X$ or $Z$ measurement of indices such that $a_i = 1$}
\begin{itemize}
    \item $\forall b \in \{0,1\}^{n'} $, \; $\trace(\sigma_Z(b) \rho) = \bra{\psi} Z(b)\ket{\psi}$
    \item $\forall b \in \{0,1\}^{n'} $, \; $\trace(\sigma_X(b) \rho) = \frac{1}{2^{n'}}\sum_a (-1)^{a\cdot b}\bra{\psi} Z(b)X(a) Z(b)\ket{\psi}$
\end{itemize}
\end{fact}

Previously we mentioned that one can retrieve samples by asking the prover measure their state in computational basis, the first bullet exactly corresponds to this scenario. Intuitively what it tells us is that the distribution of the $b(x)$, for $x$ values returned by $\prov$, is identical to the distribution of the measurement outcomes of the extracted qubit $\rho$ in the computational basis. However, in the case of the Hadamard basis, the matter is more subtle as the distribution is "twirled". As long as we only care about collecting samples via Z measurements, the twirl operator does not cause us any issues, as we will show that it would not affect the energy estimate in the protocol. 

In order to follow the proofs done in \cite{mahadev} we require our function family $\mathcal{F}$ to have an stronger property than the adaptive hardcore bit property; namely for it to be a \textit{collapsing} family, defined in \cite{DBLP:conf/eurocrypt/Unruh16}. 

Let $\mathcal{F} = \{f_{pk}\}$ be a family of functions and consider the following game,
\begin{enumerate}
    \item The challenger picks $pk \gets Gen(1^\lambda)$ and sends it to the adversary
    \item The adversary prepares a state $\ket{\phi} = \sum_x \alpha_x\ket{x}$ and sends it to the challenger
    \item The challenger evaluates $f_{pk}$ in super position of $\ket{\phi}$
    \item The challenger measures the image register and obtains $y$ and a state $\ket{\phi'}= (\sum_{x:f_{pk}(x) = y} \alpha_x \ket{x})\ket{y}$
    \item The challenger flips a bit $b$ and based on that either measures the first register of $\ket{\phi'}$ in the computational basis or does not
    \item The challenger sends the state after step 5 to the prover (either $\ket{\phi'}$ or a (classical) probabilistic mixture $\sum_{x, f_{pk}(x) = y} |\alpha_x|^2 \ket{x}\bra{x}$) 
    \item The adversary outputs a bit $b'$ based on the state he has received. 
    \item If $b = b'$ the challenger \textsf{accepts}
\end{enumerate}
$\mathcal{F}$ is called a collapsing family if for any $QPT$ adversary $\adv$, the probability of $\adv$ winning in the aforementioned game is at most $1/2 + ngl(\lambda)$. 

We proceed by stating and proving the completeness of the protocol. 

\subsubsection{Completeness}
In this section we describe an honest prover strategy. We describe a prover $\prov^{\oracle(\mathcal{D})}$ that wins in protocol~\ref{fig:classical_protocol} with probability $1$, and provides us with samples from $\mathcal{D}$. Recall that in the constant quantum memory protocol, $\prov$ first creates a history state $\ket{\phi_\dist}$ for $\ket{\psi_\dist}$ and then sends $\ket{\phi_\dist}$ to $\ver$. This prover satisfies the completeness property. For the classical model we will show that a prover who commits to the same history state also satisfies completeness.

\begin{theorem}[Completeness]
There exists a QPT prover $\prov^{\oracle(\mathcal{*})}$, such that for any distribution $\mathcal{D} \in \mathfrak{D}(n)$, any $\lambda$-collapsing claw-free family $\mathcal{F}$, $\prov^{\oracle(\mathcal{D})}$ wins in Protocol~\ref{fig:classical_protocol} with probability 1 and we have: 
\begin{itemize}
    \item $S \sim \mathcal{D}^{|S|}$
\end{itemize}
\end{theorem}

We note that completeness for this protocol is in some sense easier to prove than the completeness of the protocols described in the previous sections. The reason for this is that the protocol does not abort when the estimates are not sattisfying to desired bounds. We proceed by describing the strategy for the honest prover and the proof of the completeness.

\begin{proof}
Let us denote $2n+T'+1$ by $n'$. The honest prover will create a state such that each bit $b_i$ would correspond to the measurement of a qubit from the history state. They extend the state with zeros in the following way.
\begin{align}
\ket{\phi_{hist}}\ket{0^{mn'}}_{X}= \sum_{b_1,\dots,b_{n'}} \alpha_{b_1,\dots,b_{n'}} \ket{b_1}\ket{0^m}\dots\ket{b_{n'}}\ket{0^m}
\end{align}
By applying QFT on the 0 registers we get the state, 
\begin{align}
    \ket{\phi'} = \frac{1}{\sqrt{2^{mn'}}}\sum_{b_1,\dots,b_{n'}} \alpha_b (\sum_{z \in \{0,1\}^{mn'}} \ket{b_1}\ket{z_1}\dots\ket{b_{n'}}\ket{z_{n'}})
\end{align}
We add a zero register to $\ket{\psi}$ and evaluate $f_{pk_i}$ on the superpositions to get the state, 
\begin{align}
     \ket{\phi''} = \frac{1}{\sqrt{2^{mn'}}}\sum_{b_1,\dots,b_{n'}} \alpha_b (\sum_{x \in \{0,1\}^{mn'}} \ket{b_1}\ket{z_1}\ket{f_{pk_1}(b_1,z_1)}\dots\ket{b_{n'}}\ket{z_{n'}}\ket{f_{pk_{n'}}(b_{n'},z_{n'})})
\end{align}

$\prov$ proceeds by measuring the image registers to get values $y_1,\dots, y_{n'}$. The state after obtaining this measurement outcome will be, 
\begin{align}\label{eq:history-computational}
    \ket{\phi_{\prov}} = \sum_{b} \alpha_{b} \ket{b_1}\ket{x_{1,b_1}}\ket{y_1} \dots \ket{b_{n'}} \ket{x_{n',b_{n'}}}\ket{y_{n'}}
\end{align}
where $(b_i,x_{i,b_i})$ is the $b_i$-labeled preimage of $y_i$ under $f_{pk_i}$. 

Upon receiving challenge $0$, $\prov$ measure the state $\ket{\phi_\prov}$ in computational basis and sends $x_1 ,\dots,x_n = (b_1,x_{1,b_1}) , \dots, (b_{n'} , x_{n',b_{n'}})$ values to $\ver$. From equation \ref{eq:history-computational} one can deduce that $b_1(x_1),\dots,b_{n'}(x_{n'})$ as in the protocol is distributed identically to the outcome of the measurement of the history state in computational basis. By construction $\prov$ always succeeds in the preimage check, hence, wins with probability 1. 

As the outcomes of measuring the $b$ registers of state $\ket{\phi_{\prov}}$ are distributed identically to the measurement outcomes of the history state $\ket{\phi_\hist}$ we have $S \sim \mathcal{D}^{|S|}$, following the completeness proof from Theorem~\ref{thm:constantmemoryverifier}.
\end{proof}

In fact the completeness can be modified so that it captures the fact that the estimates computed in the protocol are close to the actual energy and outcome probability. We state the theorem here, but as similar statements are proven in the soundness we avoid repeating the proof here. 

\begin{theorem}[Completeness 2]
There exists a QPT prover $\prov^{\oracle(\mathcal{*})}$, such that for any distribution $\mathcal{D} \in \mathfrak{D}(n)$, any family of $\lambda$-collapsing claw-free family $\mathcal{F}$, $\prov^{\oracle(\mathcal{P})}$ wins $N = O(\frac{K}{\eta^4} \poly(n,T)\log(1/\delta))$ iterations of Protocol~\ref{fig:classical_protocol} with probability 1, we have that with probability at least $1-\delta$,
\begin{itemize}
    \item $|S| \geq \Omega(K)$,
    \item $p/n_2 \geq 1-2\heli^2(\mathcal{D}, \mathcal{D}_C)$,\footnote{This is done similar to the protocol described in Figure~\ref{fig:low_mem_protocol}}
    \item $\gamma/n_1 \binom{n'}{2} \in [q - \frac{\eta^2}{2T'^3} , q+ \frac{\eta^2}{2T'^3}]$,  where $q = \bra{\phi_\hist} H_\compcirc \ket{\phi_\hist}$,
    \item $S \sim (\mathcal{D})^{|S|}$.
\end{itemize}

\end{theorem}

\subsubsection{Soundness}
Now that we have established an honest prover strategy, the only thing left is to prove that for any prover who wins the game with a high probability, the verifier $\ver$ would collect samples from a distribution close to the $\dist_C$. 

The key fact to prove the soundness of the protocol is that the values $x_i$ and $d_i$ are somewhat correlated with the measurement outcomes of the $i^{th}$ qubit of the prover's extracted qubit. 

\begin{fact}[\cite{mahadev}]\label{fact:mahadev_isometry}
Let $\mathcal{F}$ be a collapsing claw-free family and let $\prov$ be any QPT prover who wins in Protocol~\ref{fig:classical_protocol} with probability 1, let $\rho$ be the prover's extracted qubits, $B_i$ and $D_i$ the outcome of measuring the $i^{th}$ qubit of $\rho$ in the computational and the Hadamard basis respectively. For any parity $\chi: \{0,1\}^{n'}\to \{-1,+1\}$ we have, 
\begin{itemize}
    \item (computational basis measurement) $\chi(B_1,\dots,B_{n'})$ is identically distributed to $\chi(b_1(x_1),\dots , b_{n'}(x_{n'}))$.
    \item (the Hadamard basis measurement) $\chi(D_1,\dots,D_{n'})$ is computationally indistinguishable from $\chi(d_1\cdot(x_{1,0} + x_{1,1}),\dots,d_{n'}\cdot(x_{n',0}+x_{n',1}))$.
\end{itemize}
\end{fact}

\begin{lemma}\label{lem:classical_expectation}
For any $T$-gate quantum circuit $C$, and its corresponding $T'$-gate comparison circuit $\compcirc$, let $\mathcal{F} = \{f_{pk}\}$ be a family of claw-free functions satisfying the collapsing property, and let $H_\compcirc$ be an $n' = 2n+T'+1$ qubit Hamiltonian corresponding to $\compcirc$. For any QPT prover $\prov$ let $\rho$ be the reduced density of the extracted qubits of $\prov$, and let $\mathcal{D}^{A}$ be the distribution of outcomes of measuring the adv register of $\rho$ conditioned on the measurement outcome of $\clock$, $\out$ and $\aux$ register being $0$ in the computational basis. If $\prov$ wins in protocol \ref{fig:classical_protocol} with probability $1$  we have,
\begin{itemize}
    \item $\mathbb{E}[\gamma/n_1 \binom{n'}{2}] \approx \trace(H_\compcirc \rho)$,
    \item $\mathbb{E}[p/n_2] = \frac{\trace(\ket{T'}\bra{T'}_{\clock}\otimes \ket{1}\bra{1}_{\out} \rho)}{\trace(\ket{T'}\bra{T'}_{\clock} \rho)}$,
    \item $S \sim {(\mathcal{D}^{A})}^{|S|}$.
\end{itemize}
\end{lemma}

\begin{proof}
To prove this theorem we will be using Fact~\ref{fact:mahadev_isometry}. First we prove the properties only relying on the computational measurements, namely properties about $p$ and $S$. Let us focus on distribution of $S$ first. 

A sample is collected if $\textsf{type} = 1$, the challenge bit $b$ is equal to $0$ and $b(o) , b(\clock) , b(\aux)$ are all 0. Due to Fact~\ref{fact:mahadev_isometry} this is equivalent to when the outcome of measuring the $\aux,\clock$ and $\out$ registers of $\rho$ are $0$. Conditioned on this happening the sample collected has the exact same distribution as measuring the adv register of $\rho$, which is equivalent to $\mathcal{D}^{A}$.

The estimate $p$ is increased by $b(o)$, when $\textsf{type} = 2$, the challenge bit is $0$ and $b(t_0),\dots, b(t_{T'}) = (1,\dots,1)$. Conditioned on the $\clock$ being $T'$, The expectation of $b(o)$ is $\trace(\ket{T'}\bra{T'}_{\clock}\otimes \ket{1}\bra{1}_{\out} \rho)$ due to fact \ref{fact:mahadev_isometry}. Hence we have $\mathbb{E}[p] = n_2\frac{ \trace(\ket{T'}\bra{T'}_{\clock}\otimes \ket{1}\bra{1}_{\out} \rho)}{\trace(\ket{T'}\bra{T'}_{\clock} \rho)}$.

The next thing to prove is that the energy estimate has the desired expectation. If we consider the $n_1$ rounds in which we change the energy estimate, the expectation of the amount of change done to $\gamma$ is equal to:
\begin{center}
$-\frac{1}{2\binom{n'}{2}}\sum_{i,j} J_{i,j}(-1^{b_i(x_i) + b_j(x_j)} + -1^{d_i(x_{i,0}+ x_{i,1}) +d_{j} (x_{j,0} + x_{j,1})})$
\end{center}
For $b_i(x_i)$ and $b_j(x_j)$, we know that these random variables are distributed identically to measurement of $\rho$ in computational basis. The only issue is that $d_i(x_{i,0} + x_{i,1})$ is not distributed identically to Hadamard measurement of $\rho$, but rather is computationally indistinguishable from it.

However for any parity $\chi$ if the distance between the expectations of $\chi(d_i(x_{i,0} + x_{i,1}))$ and $\chi$ applied on the measurement outcomes of $\rho$ in the Hadamard basis is negligible in $\lambda$; as otherwise an adversary could distinguish between the two by random sampling using only $O(1/poly(\mu))$ samples. Hence we have, 
\begin{align*}
   \mathbb{E}[J_{i,j}(-1)^{b_i(x_i)+ b_j{x_j}}] = J_{i,j} \trace(\sigma_{Z,i} \sigma_{Z,j}\rho) &\\
  \mathbb{E}[J_{i,j}(-1)^{d_i(x_{i,0}+ x_{i,1}) +d_{j} (x_{j,0} + x_{j,1})}] = J_{i,j}(\trace(\sigma_{X,i} \sigma_{X,j}\rho) \pm ngl(\lambda))&
\end{align*}
Hence we have that $\mathbb{E}[\gamma] \approx n_1 \frac{1}{\binom{n'}{2}} \trace(H_\compcirc \rho)$ as desired. 
\end{proof}

\begin{theorem}[Perfect Prover Soundness]\label{thm:classical}
For any $T$-gate circuit $C$ acting on $n$-qubits, the protocol defined in Figure~\ref{fig:classical_protocol} has the following properties. It is an interactive protocol $(\ver,*)$ between a classical verifier $\ver$ and a quantum prover. For any QPT prover $\prov$, let $\rho$ be the reduced density of the extracted qubits of  $\prov$, and $\mathcal{D}^{A}$ be the distribution of outcomes of measuring the adv register of $\rho$ in the computational basis conditioned on the measurement outcome of $\clock$, $\out$ and $\aux$ registers being $0$. If $\prov$ wins $N = O(\frac{K}{\eta^4}\poly(n,T)\log(1/\delta))$ itterations of the protocol with probability $1$ and $\frac{\gamma}{n_1} \binom{2n + T+1}{2} \leq \frac{\eta^2}{2T'^3}$ and $\frac{p}{n_2} \geq 1-2\eta^2$, then with probability $1-\delta$ we have,
\begin{itemize}
    \item $\heli(\mathcal{D}_C , \mathcal{D}^{A}) \leq O(\eta^{1/4})$,
    \item $S \sim (\mathcal{D}^A)^{|S|}$.
\end{itemize}

\end{theorem}

\paragraph{Note.}
The guarantee expressed in the last sentence of Theorem~\ref{thm:classical} might seem mysterious at first. Note however that the conditions contained there are equivalent to the checks performed at the last step in the constant quantum memory protocol from Figure~\ref{fig:low_mem_protocol}. The fact that the checks are contained in the statement of the theorem and not in the protocol itself allows us to analyze perfect provers only and simplifies the presentation considerably. 

\begin{proof}[Proof of Theorem~\ref{thm:classical}]
Let $\compcirc$ be the $T'$ comparison circuit of $C$ and let $H_\compcirc$ be the corresponding 2-local Hamiltonian acting on $n' = 2n + T' +1$ qubits. 

Applying Lemma~\ref{lem:threets} we have that with probability $1- e^{-\frac{N}{18}}$ we have $n_1 \geq \Omega(N)$. From Lemma~\ref{lem:classical_expectation} we have,
\begin{align}
    \mathbb{E}\left[\frac{\gamma}{n_1} \binom{n'}{2}\right] \approx \trace(H_\compcirc \rho)\\ 
    \mathbb{E}\left[\frac{p}{n_2} \right]= \frac{\trace(\ket{T'}\bra{T'}_{\clock}\otimes \ket{1}\bra{1}_{\out} \rho)}{\trace(\ket{T'}\bra{T'}_{\clock} \rho)}
\end{align}

Using Fact~\ref{fact:chernoff} we get, 
\begin{align}
    \Prob \left[ \left|\frac{\gamma}{n_1} \binom{n'}{2} - \trace(H_\compcirc \rho)\right| \geq \frac{\eta^2}{2T'^3} \right] \leq 2e^{-\frac{\eta^4 n_1}{8 \binom{n'}{2}^2 T'^6 J}} \label{gamma_inq}\\
    \Prob \left[\left|\frac{p}{n_2} - \frac{\trace(\ket{T'}\bra{T'}_{\clock}\otimes \ket{1}\bra{1}_{\out} \rho)}{\trace(\ket{T'}\bra{T'}_{\clock} \rho)}\right| \geq \frac{3\eta^2}{2} \right] \leq 2 e^{-\frac{9\eta^4 n_2}{8}} \label{p_inq},
\end{align}
where $J = \text{sup}_{i\neq j \in [n']} \{|J_{i,j}|\}$.

If we use the hypothesis of the theorem, \eqref{gamma_inq} and \eqref{p_inq} we get that with probability $1 - \frac{\delta}{8}$,
\begin{align}
    \trace(H_\compcirc \rho) \leq \frac{\eta^2}{T'^3}\\
    \frac{\trace(\ket{T'}\bra{T'}_{\clock}\otimes \ket{1}\bra{1}_{\out} \rho)}{\trace(\ket{T'}\bra{T'}_{\clock} \rho)} \geq 1- 2\eta^2 - \frac{3\eta^2}{2} \geq {1-\frac{7\eta^2}{2}} \label{output_bit}
\end{align}
 Note that \eqref{output_bit} implies that probability of the measurement outcome of the out register being $1$, when the $\clock$ is $T'$ is at least $1-\frac{7\eta^2}{2}$. By employing Corollary~\ref{cor:circuittoham} we have that with probability $\frac{1\pm 7\eta}{T'+1}$, $d_H(\mathcal{D}_C , \mathcal{D}^{A}) \leq O(\eta ^{1/4})$. 

By Fact~\ref{fact:chernoff} we have that $n_2  = \Omega(\frac{N}{T'})$ with probability $1 - \frac{\delta}{20}$ so it's enough for $N = O( \frac{1}{\eta^4}\log(1/\delta)\binom{n'}{2}^2 T'^7 J) = O(\frac{1}{\eta^4}\log(\frac{1}{\delta})\poly(T,n))$ for \eqref{p_inq} to hold with probability $\leq \delta/10$. If we apply the union bound over all failure events we get that all the conditions will be satisfied with probability at least $1-\delta$, hence with probability $1-\delta$ we get $\heli(\mathcal{D}_C ,\mathcal{D}^{A}) \leq O(\eta^{1/4})$.

The second bullet follows directly from Lemma~\ref{lem:classical_expectation}.
\end{proof}

\paragraph{Discussion}
We note that we have proven the soundness of our protocol only in the perfect prover setting. The problem with this statement is that, no verifier can make sure that the prover is winning with probability $1$. Also the soundness guaranty is different from the previous sections as the game does not \textsf{abort} when the estimates do not satisfy the bound. The reason we modified the game in this manner is that if the game aborted after checking the bounds, even the honest prover would not have won the game with probability $1$.

However, it is possible to achieve a stronger soundness guaranties, similar to Theorem~\ref{thm:constantmemoryverifier}. To do so, we would need to tweak the protocol, namely by utilizing yet another commitment called a trapdoor injective family. To prove soundness of a non-perfect prover one can then follow a similar path as the one in Claim 7.1 of \cite{mahadev}. In there a reduction from the non-perfect prover to a perfect prover is given.




\section{Future Work}

We show how to reduce the problem of constructing adversarially robust classifiers to the problem of finding generative models plus the problem of devising classifiers that are robust with respect to distributional shifts. This firmly places the problem of adversarial robustness in the landscape of learning theory. Our result is valid in a model where the prover has access to quantum samples. 

There are some important open questions that we leave for future research.
First, it would be of great interest to find real-world applications where the underlying model is indeed of quantum nature and hence this model applies directly. Second, 
can we use the insights from our results to design defenses that are valid in the classical PAC-learning model? As we discussed in Section~\ref{sec:overview}, this might be challenging because certifying randomness between entities that are fully classical seems unreachable in our setup. But it is possible that there exists a third model, neither PAC nor quantum PAC, in which it is possible to obtain a result of similar flavor. Moreover, we can turn the question around and ask for a reduction of the problem of certifiable randomness to the adversarial robustness problem. If such a reduction exists then this strongly implies that in order to tackle adversarial robustness we should first solve the problem of certifiable randomness. Finally, our results are valid for the Hellinger distance. It is natural to ask if our reduction can be generalized to other similarity metrics and if such a potential generalization might be advantageous.

\section*{Acknowledgment}
We would like to thank Thomas Vidick, whose superb lecture notes \citep{VidickLectureNotes} have aided us substantially in learning about the recent progress in quantum protocols.

\bibliography{example_paper}
\bibliographystyle{abbrvnat}





\end{document}